\documentclass[11pt]{article}
\usepackage{amsmath,amssymb}
\usepackage[left=1in,right=1in,top=1.1in,bottom=1in]{geometry}
\usepackage{bm}
\usepackage{expl3}

\usepackage{tikz}
\usepackage{appendix}
\usepackage{subfigure}
\usepackage[english]{babel}
\usepackage{bbm}
\usepackage{graphicx}
\usepackage[center]{caption2}
\usepackage{amsfonts,amssymb,amsmath,latexsym,amsthm}
\usepackage{multirow}
\usepackage{enumerate}
\usepackage{geometry}
\usepackage[usenames,dvipsnames]{pstricks}
\usepackage{epsfig}
\usepackage{pst-grad} % For gradients
\usepackage{pst-plot} % For axes
\usepackage[space]{grffile} % For spaces in paths
\usepackage{enumitem}% delete small point before items
\usepackage{etoolbox} % For spaces in paths
\makeatletter % For spaces in paths
\patchcmd\Gread@eps{\@inputcheck#1 }{\@inputcheck"#1"\relax}{}{}

\newtheorem{thm}{Theorem}[section]
\newtheorem{cor}[thm]{Corollary}
\newtheorem{lem}[thm]{Lemma}

\newtheorem{claim}[thm]{Claim}

\newtheorem{defn}[thm]{Definition}

\newtheorem{bq}[thm]{Main Question}

\newtheoremstyle{plainupright}
  {\topsep}   % Space above
  {\topsep}   % Space below
  {\upshape}  % Body font
  {}          % Indent amount
  {\bfseries} % Theorem head font
  {.}         % Punctuation after theorem head
  { }         % Space after theorem head
  {}          % Theorem head spec (can be left empty, meaning 'norma\lan')
\theoremstyle{plainupright}
\newtheorem{example}{Example}
\newtheorem{remark}{Remark}

\let\svthefootnote\thefootnote
\newcommand\blankfootnote[1]{%
	\let\thefootnote\relax\footnotetext{#1}%
	\let\thefootnote\svthefootnote%
}

\newcommand{\xhdr}[1]{\paragraph{\bf #1.}}

\newcommand{\omt}[1]{}

\newcommand\supproof[1]{} 
\newcommand\topo[1]{\tau_{#1}}

\def\coll{{\mathcal{X}}}

\def\trueL{{K}}

\newcommand{\Lc}[2]{L_{{#1}}^{({#2})}}
\newcommand{\Ic}[2]{\mathcal{I}_{{#1}}^{({#2})}}

\newcommand{\mX}{{\mathcal{X}}}
\newcommand{\ma}{{\mathcal{A}}}

\newcommand{\mP}{{\mathcal{P}}}

\newcommand{\mB}{{\mathcal{B}}}
\newcommand{\Acc}{\mathcal{A}_{\text{acc}}}

\newcommand{\mS}{{\mathcal{S}}}

\newcommand{\Succ}{\text{succ}}

\newcommand{\I}[1]{\mathcal{I}^{(#1)}}
\newcommand{\mCh}[1]{\mathcal{C}_{#1}}

\newcommand{\lan}{\mathsf{L}}
\newcommand{\lanj}[1]{\mathsf{J}_{#1}}

\newcommand{\fbga}[1]{\widetilde{\mathcal{I}}^{(#1)}}

\def\groundset{\mathcal{Y}}

\ExplSyntaxOn

\ExplSyntaxOff

\begin{document}

\pagenumbering{gobble}
	
\title{Language Generation and Identification From Partial Enumeration: Tight Density Bounds and Topological Characterizations}
	
\date{\today}

 % \author{Authors anonymized for submission version}
  \author{Jon Kleinberg\thanks{Department of Computer Science and Information Science, Cornell University, Ithaca NY 14853 USA.  Supported in part by a Vannevar Bush Faculty Fellowship, AFOSR grant FA9550-23-1-0410, a Simons Collaboration grant, and a grant from the MacArthur Foundation.} \and Fan Wei\thanks{Department of Mathemaics, Duke University, 120 Science Drive, Durham, NC 27710, USA. Research supported by NSF grants DMS-2404167 and DMS-2401414. } }

\maketitle

\begin{abstract}
The recent successes of large language models (LLMs) have led to active lines of work in formal theories of language generation and learning. We build on one such theory, {\em language generation in the limit}, in which an adversary enumerates the strings of an unknown language $K$ drawn from a countable list of candidate languages, 
and an algorithm tries to generate unseen strings from the language. 
Initial work on this model showed there is an algorithm that can always succeed at this task, and more recent work has shown there is in fact an algorithm that can produce a positive-density subset of the language. These results on density reflect the \emph{validity--breadth} tension in language generation: the trade-off between generating only valid strings while also achieving wide coverage of the true language. 

Here we begin by resolving one of the main open questions from this work on density, establishing a tight bound of $1/2$ on the best achievable lower density of any algorithm. We then consider a more powerful adversary, capturing the fact that generation algorithms may typically be faced with an environment in which only a subset of the language is being produced. This is a model with only {\em partial enumeration of $K$}: We show that there is an algorithm with the property that if an adversary only outputs an infinite subset $C$ of the true language $K$, it can still achieve language generation in the limit; and moreover, if the subset $C$ has lower density $\alpha$ in $K$, then the algorithm produces a subset of lower density at least $\alpha/2$, which matches the upper bound. This generalizes the tight density bound of $1/2$ to the case where the algorithm must come within $1/2$ of the density of whichever subset of $K$ the adversary reveals. 

We also revisit the classical Gold-Angluin model of language identification (rather than generation) when the adversary need only partially enumerate an infinite subset $C$ of the true language $K$.
We characterize when it is possible for an algorithm to achieve the natural analogue of identification in the limit in this partial setting, producing languages $M_t$ (and finite representations of them) such that eventually $C \subseteq M \subseteq K$.  Our characterization builds on our earlier topological approach  on density in language generation \cite{kleinberg2025density}, and in the process we give a new topological formulation of Angluin's characterization for language identification in the limit, showing that her condition is precisely equivalent to some appropriate topological space having the $T_D$ separation property. 
\end{abstract}

\newpage

\tableofcontents

\newpage

\pagenumbering{arabic}
\setcounter{page}{1}

\section{Introduction}
The rapid growth of generative AI has made \emph{language generation}
central to modern computing. 
Large language models (LLMs) trained on text corpora
consistently produce new, valid text beyond their training examples,
and understanding how they do this has been the subject of
several important lines of theoretical analyses, ranging from 
studies of trasnformer-based architectures that dominate
current LLM design
\cite{peng-transformer,sanford-transformer,wang-transformer,weiss-transformer}
to more abstract models that formulate language generation
as a standalone computational problem
\cite{charikar-pabbarju,kalai2023calibrated,kalavasis-stoc25,
kleinberg2024limit, kleinberg2025density,li-generation-pac}.

In this paper, we work within this second type of model, adopting a {\textbf{basic,
assumption-free}} approach to study language generation, in a 
line of research that began with work by Gold and Angluin
on language identification in the limit
\cite{angluin1979finding,angluin1980inductive,gold1967language}.
In the Gold-Angluin model, 
there is a countable collection of languages
$\mathcal{X} = \{L_1, L_2, L_3, \dots\}$, each of which is an infinite
subset of a ground set $\groundset$ consisting of all finite-length strings over
some alphabet.
An adversary chooses a language $K \in \mathcal{X}$ and enumerates
all the strings of $K$ in some order, one in each time step.
An algorithm watches these strings one-by-one, and in each time step
the algorithm guesses an index $i_t$.
The algorithm succeeds --- it achieves
{\em language identification in the limit} --- if after
some finite time $t^*$ its guess is correct forever;
that is, if $L_{i_t} = K$ for all $t \geq t^*$.
Gold proved that no algorithm can succeed at this problem
even for very simple collections $\mathcal{X}$
\cite{gold1967language}, and
Angluin subsequently characterized the collections for which
an algorithm can succeed \cite{angluin1979finding,angluin1980inductive}.

Recently, a new type of question was introduced by
Kleinberg and Mullainathan \cite{kleinberg2024limit}
into this set of models,
motivated by developments in large language models:
in this new question, the goal is 
{\textbf{not identification but generation}}.
Specifically, there is still an adversary that 
chooses a language $K \in \mathcal{X}$ and enumerates
all the strings of $K$ in some order, one in each time step.
But now the algorithm that watches these strings one-by-one
is not required to identify the language, only to 
{\em generate} from it:
in step $t$, it must output a string $a_t$ from $K$ that the adversary
has not yet produced.
In other words, the algorithm seeks to satisfy the requirement
that $a_t \in K - S_t$, where $S_t$ is the set of strings
produced by the adversary in the first $t$ steps.
In this new model, the algorithm achieves 
{\em language generation in the limit}
if there is some time $t^*$ such that 
$a_t \in K - S_t$ for all $t \geq t^*$.
And the main result of Kleinberg and Mullainathan is that there
is an algorithm (which we will call the {\em KM algorithm}) that
achieves generation in the limit for {\em every}
countable collection of languages $\mathcal{X}$.
This is very different from the Gold-Angluin impossibility results, and
shows that generation in the limit is much more tractable than
identification in the limit.

\xhdr{Tension between Breadth and Validity in Language Generation}
An active line of research on language generation in the limit followed
up quickly on the initial Kleinberg-Mullainathan results
\cite{kalavasis-stoc25, kalavasis2024breadth, kalavasis2025limits, charikar-pabbarju, bai2025noise, raman2025noisy, li2025learning, hanneke2025union}. 
The portion of this work that is most relevant for our purposes
is concerned with the problem of {\em breadth}:
how ``much'' of the unknown language $K$ is an algorithm able to generate?
This question mirrors an issue of central important in applied
work on LLMs and language generation, which seeks to balance between
avoiding {\em hallucination} at one extreme \cite{kalai2023calibrated} --- 
in which
the generated strings are not valid members of the true language --- 
and avoiding {\em mode collapse} at the other extreme --- in which 
the set of outputs are very sparse relative to the full set of valid outputs
\cite{allen2024physics}.

The original KM algorithm generated a subset of $K$ that, while infinite,
could be very sparse in a sense we will make precise shortly.
It therefore raised the question of whether there was room for
improvement on its breadth, or whether this type of mode collapse 
was unavoidable, as asked in the paper by Kleinberg and Mullainathan \cite{kleinberg2024limit}. One might guess that this tension is inevitable. 
To understanding this question better, it is useful to observe that the original KM algorithm has the following property:  it generates strings by keeping track of a descending chain of candidate languages, each contained within the previous one, all consistent with observed data. In striving to avoid overshooting the true language in order to guarantee validity, is forced to descend down the chain continuously. It ensures correctness but at the cost of density: each update prunes the space further, so even though the algorithm eventually reaches subsets of the true language, its outputs occupy an asymptotically vanishing fraction of it.

Several follow-up papers pointed in the negative direction by
showing inherent limits on how much of $K$ could be generated.
In this direction, 
Charikar and Pabbaraju~\cite{charikar-pabbarju} showed an remarkable negative result that
there are instances where if the adversary stops showing further
samples from $K$ at any finite time, no algorithm can output the
entire rest of $K$, and Bai et al showed similar striking negative result even if the adversary only insert a finite noisy inputs \cite{bai2025noise}. 
Kalavasis et al.~\cite{kalavasis2024breadth} studied a statistical model
where strings from \( K \) are sampled independently from a fixed
distribution \( \mathcal{P} \), and the adversary reveals strings in
an average-case fashion; they showed limitations on the probability
any algorithm could achieve in trying to generate unseen strings. 

Our research aim to study this problem in a \textbf{worst-case} setting, rather than under average-case or probabilistic assumptions, under some density notion. 
As the first work among this direction, as a counterweight to these negative results, we  \cite{kleinberg2025density} have
 developed a framework for establishing
positive results for breadth, and for quantifying achievable levels
of breadth numerically, via {\em density measures} 
commonly used in analytic number theory and combinatorics.
The basic premise was the following:
\begin{quote} 
While it may be impossible to generate \emph{all} unseen strings, it is still meaningful—and potentially tractable and more practical—for an algorithm to achieve \textbf{nontrivial density} over the target language. 
\end{quote}

Concretely, in \cite{kleinberg2025density}, we defined breadth
by the \textit{asymptotic density} of strings generated by the algorithm within the true language. The density captures the fraction of relevant strings produced over increasingly large subsets of language instances, offering a precise way to measure the breadth of generative behavior:

\begin{defn}
    If $L$ and $L'$ are two languages, and the elements of $L'$ listed in order as
$L' = \{\ell_1', \ell_2', \ell_3', \ldots\}$,
we say that the {\em upper density} of $L$ in $L'$ is
   \[ d_{\text{up}}(L, L') = \limsup_{N \to \infty}  \frac{\mid L \cap \{\ell_1', \dots, \ell_N'\} \mid}{N}.\]
Analogously, we say that the {\em lower density} of $L$ in $L'$ is
   $d_{\text{low}}(L, L') = \liminf_{N \to \infty}  \frac{\mid L \cap \{\ell_1', \dots, \ell_N'\} \mid}{N}.$
\end{defn}

It is necessary to have both definitions: as shown in \cite{kleinberg2025density},
the upper and lower densities can differ dramatically if, for example,
$L$ consists
of increasingly long intervals that alternate between containing
many elements of $L'$ and then very few elements of $L'$.
There are many examples where
the upper density can be arbitrarily close to one while the lower density is arbitrarily close to zero.

Finally, to apply this definition, let  
$O(E,\ma)$, for an adversarial enumeration $E$ of a language
$K \in \mathcal{X}$ and a generation algorithm $\ma$,
 be the set of all strings that $\ma$ outputs, over all time steps,
when presented with the enumeration $E$.
For the original KM algorithm, there are instances where
the algorithm's output set $O(E,\ma)$ has upper density zero
(and hence lower density zero) 
in the true language $K$ --- this is the sense in which it can
generate an arbitrarily sparse subset of $K$.
It was therefore natural to ask whether zero density might 
be unavoidable for any algorithm that achieves generation in the limit.
But we showed that this is not the case:
our paper \cite{kleinberg2025density} gave an algorithm achieving
generation in the limit such that for some
absolute positive constant $c > 0$, the algorithm generates a set of
strings $O(E,\ma)$ of lower density at least $c$ in all instances 
specified by any collection of languages $\mathcal{X}$,
and any enumeration $E$ of one of the languages in $\mathcal{X}$
\cite{kleinberg2025density}.
Their work thus provides theoretical support to the idea that,
despite the impossibility results of 
\cite{kalavasis2024breadth,charikar-pabbarju},
intermediate guarantees are possible,
aligning with the empirical evidence that in practice, 
models can be trained to simultaneously reduce hallucination 
and alleviate mode collapse \cite{allen2024physics,arjovsky2017towards,arjovsky2017wasserstein}. 

A fundamental question arising from \cite{kleinberg2025density} was then the following:
\begin{bq}\label{bq:tension}
What is the best constant $c > 0$
such that some algorithm $\ma$ achieves generation in the limit,
and for every instance 
the lower density of the algorithm's output set $O(E,\ma)$ in $K$ is always
at least $c$?
\end{bq}
We were able to achieve a constant of $c = 1/8$ in \cite{kleinberg2025density}, but
the best known upper bound was $c = 1/2$ (which follows in a straightforward
way from the fact that an adversary can output roughly half the
strings before the algorithm can get to them, even if the algorithm
knows the identity of the true language $K$).
It was therefore a basic open question to determine whether the
natural upper bound of $c = 1/2$ was in fact the tight bound. This problem was also highlighted as a main open question in~\cite{kleinberg2025density}.

\subsection{The Present Work: Tight Density Bounds and Partial Enumeration}

Our first contribution is to completely resolve Question \ref{bq:tension},
by showing that the upper bound $c = 1/2$ is indeed tight:

\begin{thm}\label{thm:introbestboundfull}
    There is an algorithm $\ma$ that achieves generation in the limit and has the following property. Given any countable collection of languages $\coll$ with  an underlying ordering of the strings in all the languages, and for any adversarial enumeration $E$ of one of the languages $\trueL \in \coll$, the set of output strings $O(E,\ma)$ generated by the algorithm $\ma$ has a lower density in $K$  at least $1/2$.
\end{thm}

In fact, we establish this result in a more general
model that addresses an important issue that has been essentially
unexplored in the emerging literature on 
language generation in the limit: the question of {\textbf{partial enumeration}}.
Specifically, the existing work on this problem makes crucial
use of the assumption that the adversary is required to enumerate
the entire language $K$, i.e., that for every string $w \in K$, 
there is a time-step $t$ such that the adversary outputs $w$ at time $t$.
But there is no a priori indication why complete enumeration should be
necessary, and in practice one of the most 
fundamental and perplexing phenomena in modern LLMs
is their remarkable ability to generalize
effectively from highly incomplete training data.

We therefore develop a formal framework for addressing the following
theoretical question:
\begin{quote}
\vspace{-0.05in}
\textit{If the adversary enumerates only an infinite  subset \(C\) of the true language \(K\), under what conditions can an algorithm still guarantee correct generation in the limit? What guarantees do we have about the output?}
\end{quote}

A model for this partial enumeration question 
offers the opportunity to abstract several 
key aspects of LLM behavior. First, LLMs often generalize from highly sparse data, and partial enumeration may help provide insight into this process.
Second, real-world training data is frequently incomplete, censored, or even adversarially selected; partial enumeration naturally allows for
adversarially selected subsets of the training data.
Finally, success under partial data suggests that sensitive examples may not be required: again keeping in mind that we are allowing arbitrary
subsets of the language to be enumerated, thus 
shedding light on the trade-off between memorization and privacy.

\paragraph{Mathematical Formulation of Partial Enumeration}
We therefore consider the following model of partial enumeration:
given a countable collection of languages $\mathcal{X}$ as before,
an adversary chooses $K \in \coll$, {\em then chooses an arbitrary
infinite subset $C$ of $K$}, and then enumerates the strings of $C$
in an arbitrary order.
An algorithm seeking to achieve generation in the limit operates
as before: it sees these enumerated strings, one in each time step,
and in time step $t$ it outputs a string $a_t$.
The goal is the same as before: for some time $t^*$, the algorithm's
outputs satisfy $a_t \in K - S_t$ for all $t \geq t^*$.
We say that the algorithm achieves 
{\em generation in the limit with partial enumeration} for the
collection $\coll$ if for any language $K \in \coll$, and
any infinite subset $C$ of $K$, and any enumeration of $C$,
the algorithm achieves the guarantee that for some time $t^*$,
its outputs $a_t$ satisfy $a_t \in K - S_t$ for all $t \geq t^*$.
The key distinction with the traditional definition is that
the adversary only needs to enumerate from the infinite subset $C$,
and thus there may exist infinitely many strings $w \in K$ 
that never appear in the enumeration.

With this definition, we can now ask our basic open question as follows:

\begin{bq}
Which classes of languages \(\mathcal{X}\) admit successful generation 
in the limit when the adversary can perform partial enumeration
(using an arbitrary infinite subset $C$ of its chosen language $K \in \coll$)?
\end{bq}

Our first main result on this model is that generation in the limit
with partial enumeration is possible for {\em every} countable collection 
of languages $\coll$:

\begin{thm}\label{thm:introKWnew}
There is an algorithm $\ma$ with the property
that it achieves generation in the limit with partial enumeration
for every countable collection of languages.
\end{thm}
In fact, we are able to prove a much stronger theorem. In Theorem \ref{thm:introaccuate}, we showed that not only can we guarantee generation in the limit, but the generation is basically ``accurate" or ``full" infinitely often. 
We will give an overview of how we prove this theorem below.
For now, we observe two interesting things about
Theorem \ref{thm:introKWnew}.
First, in contrast to the adversary from the standard model
of generation in the limit, who could choose a language to generate
from the countable collection $\coll$, 
the adversary in the model of partial enumeration can choose
from among the much larger {\em uncountable} collection of subsets of the
languages in $\coll$.
This is a concrete sense in which the adversary in the partial enumeration
model has much more freedom.
Second, the algorithm $\ma$ achieving the guarantee in 
Theorem \ref{thm:introKWnew} does not need to have its output
set $O(E,\ma)$ be a subset of the adversary's infinite subset $C \subseteq K$;
indeed, it does not for $O(E,\ma)$ to intersect $C$ at all, as long as
$O(E,\ma)$ is a subset of $K$.
And this is necessary, since in general the algorithm cannot guarantee
anything about the intersection of $C$ and $O(E,\ma)$.

We next combine the model of partial enumeration with the 
formulation of lower density to arrive at a result that combines
the conclusions of Theorems \ref{thm:introbestboundfull} 
and \ref{thm:introKWnew}:

\begin{thm}\label{thm:intropartialdensity}
    There is an algorithm $\ma$ that achieves generation in the limit with
partial enumeration and that has the following property. 
Given any countable collection of languages $\coll$ with an
underlying ordering of the strings in all the languages, and for any
adversarial enumeration $E$ of an infinite subset $C$ of
one of the languages $\trueL \in \coll$ such that $C$ has
lower density at least $\alpha > 0$ in $K$,
the set of output strings $O(E,\ma)$ generated by the algorithm has a lower density in $K$ that is at least $\alpha/2$.
\end{thm}
Again, $\alpha/2$ also serves as an upper bound on the lower density
achievable by any algorithm,
analogous to the situation in the full enumeration model. Note that
Theorem~\ref{thm:introbestboundfull} is a special case of
Theorem~\ref{thm:intropartialdensity}.

\subsection{Index-Based and Conjunction-Based Generation}

In order to discuss the approach to proving these results
on partial enumeration, it is useful to start by going back
to the definition of generation in the limit.
In particular, we will refer to the basic guarantee that
$a_t \in K - S_t$ for all $t \geq t^*$ as 
{\em element-based generation in the limit}, to contrast it
with other notions we introduce now.

The algorithms for generation in the limit in previous work,
beginning with the KM algorithm, typically 
achieve a stronger guarantee called
{\em index-based generation in the limit}:
they not only output a string $a_t$ in time step $t$, they
also guess an index $i_t$ for one of the languages in $\coll$
with the property that for some time $t^*$ and all $t \geq t^*$,
we have $L_{i_t} \subseteq K$.
Any algorithm that achieves index-based generation in the limit 
can also achieve element-based generation in the limit:
by outputting an unseen string from $L_{i_t}$, it is guaranteed
to be outputting a string in $K - S_t$, since $L_{i_t} \subseteq K$.

In the case of an adversary that does partial enumeration, however, 
we discover a surprising separation between element-based generation
and index-based generation: it is possible to achieve only the former,
but not always the latter.
To illustrate, consider a family of languages $\{L_n\}_{n \in
\mathbb{N}}$, where each $L_n$ contains all natural numbers except
$n$. If the true language $K$ is $L_z$ for some odd number $z$, and
the adversary only enumerates even numbers, an algorithm has no way to
correctly identify an index $i_t$ such that $L_{i_t} \subseteq K$.
This distinction is significant because it highlights a fundamental
difference between element-based and index-based generation, a contrast that is
absent in the standard model where the adversary enumerates all of $K$
rather than just an infinite subset.

Is there some weaker concept that we can use in place of 
index-based generation?
We show that in fact there is, via finite intersections of
languages in $\coll$.
We say that an algorithm $\ma$ achieves {\em conjunction-based}
generation in the limit with partial enumeration
if in each time step $t$, the algorithm chooses a finite set of
indices $i_t(1), i_t(2), \ldots, i_t(r_t)$, and for some time $t^*$,
the set 
$$M_t = L_{i_t(1)} \cap L_{i_t(2)} \cap \cdots \cap L_{i_t(r)}$$
has the property that $M_t \subseteq K$ for all $t \geq t^*$.
(We will also refer to conjunction-based generation as 
{\em semi-index-based generation}.)
If an algorithm achieves conjunction-based generation in the limit, it
can also achieve element-based generation in the limit, simply by
outputting a string in $M_t$ at time step $t$.

We prove Theorem \ref{thm:introKWnew} by establishing the
following stronger result:

\begin{thm}\label{thm:introConjunction}
There is an algorithm $\ma$ with the property
that it achieves conjunction-based
generation in the limit with partial enumeration
for every countable collection of languages.
\end{thm}

In fact we can design an algorithm with a stronger guarantee,
that the finite intersection $M_t$ it produces contains 
the adversary's enumerated set $C$ infinitely often.
This stronger guarantee is an important building block in
the positive-density result of 
Theorem \ref{thm:intropartialdensity}.

\begin{thm}\label{thm:introaccuate}
There is an algorithm $\ma$ with the property
that it achieves conjunction-based
generation in the limit with partial enumeration
for every countable collection of languages, and
the set $M_t$ produced by the algorithm in step $t$
has the property that if $C$ is the infinite subset of $K$
enumerated by the adversary, we have $C \subseteq M_t$
for an infinite set of time steps $t$.
\end{thm}

\subsection{Identification from Partial Enumeration, and Topological Connections}

Once we have a framework for reasoning about adversaries that
perform partial enumeration, it is interesting to shift
from generation back to identification, and to ask whether
we can say anything novel about the classical Gold-Angluin problem
of language identification in the limit when we add the
possibility of partial enumeration.

In the standard formulation of identification in the limit,
where the algorithm outputs the index $i_t$ of a language in $\coll$,
with the goal that $L_{i_t} = K$ for all $t \geq t^*$,
the possibility of partial enumeration makes it hard to say
anything non-trivial.
But when we allow for the possibility of a conjunction-based
algorithm --- one that in step $t$ 
outputs a finite intersection $M_t$ of languages 
in $\coll$ --- we can begin to provide much more interesting guarantees.

In particular, we would like to ask when the following 
type of identification is possible by a conjunction-based algorithm:
\begin{quote}
When the adversary enumerates the strings of an infinite subset $C$
of $K \in \coll$, there should be some time $t^*$ such that
for all $t \geq t^*$, the algorithm's set $M_t$ contains $C$ and is
contained in $K$.  That is, $C \subseteq M_t \subseteq K$ for all
$t \geq t^*$.
\end{quote}
We will refer to this as {\it conjunction-based identification in the
limit with partial enumeration}.
We observe that it is a generalization of traditional identification
in the limit, since in the special case when the adversary
enumerates the full language $C = K$, our requirement 
that $C \subseteq M_t \subseteq K$ becomes $M_t = K$:
the algorithm must output exactly $K$ after some finite time.

In order to express our characterization of when this is possible,
we extend some of the
topological definitions used in our 
analysis of generation in the limit with positive density
\cite{kleinberg2025density}; 
we develop the definitions here for the context of identification
rather than generation.

\xhdr{Topological definitions}
The paper~\cite{kleinberg2025density} was the first to draw a
connection between language generation in the limit
and topological constructions,
interpreting the generation process as the search for a limit
object within a carefully defined topological space under the full
enumeration model.

Without loss of generality, we assume that each language in $\mathcal{X}$ is a subset of $\mathbb{N}$
(since any discrete countable set can be enumerated by $\mathbb{N}$).
Given a collection of languages $\mathcal{X}$, we define a topology $\mathcal{T}$ on it as follows.
For each language $L \in \mathcal{X}$ and each finite subset $F \subseteq \mathbb{N}$,
define a basic open set by
\begin{equation}
U_{L,F} = \{L' \in \mathcal{X} \mid F \subseteq L' \subseteq L\}. \label{eq:base}
\end{equation}
Let the collection of sets \( U_{L,F} \), ranging over all \( L \in \mathcal{X} \) and finite \( F \subset \mathbb{N} \),
serve as a basis for the open sets; this defines a topology on \( \mathcal{X} \).
This topology has a key feature: its notion of limit points corresponds closely to the underlying
obstructions that cause the algorithm to become trapped in low-density output regions.

In~\cite{kleinberg2025density}, it was shown that when $(\mathcal{X}, \mathcal{T})$ has an empty perfect kernel
and finite Cantor--Bendixson rank~$r$, there exists an algorithm in the full enumeration model that
identifies in the limit with lower output density at least $1/(3(r+1))$.
(Note our new Theorem~\ref{thm:introbestboundfull} improves this bound and achieves the optimal guarantee.)

\xhdr{Topological characterizations}
We now discuss how we define topological objects that help in
characterizing identification rather than generation.
The first striking finding is that we can take Angluin's celebrated
characterization of the language collections for which 
identification in the limit is possible, and we can give a new
formulation of it in terms of a classical topological separation property.

Let us use $\mathcal{T}_{\rm{full}}$ to denote the topological space where
the open sets are generated by the basis $U_F = \{\lan \mid F \subset
\lan\}$ for finite $F$.   
(We use ``full'' as the subscript to denote that this is the relevant topological space for language identification when the adversary enumerates the full language $K$ rather than partially enumerating a subset.)
% The lower index is in honor of Gold and Angluin. 
The language identification problem is closely related to
separation properties of the space, i.e., how can we distinguish two
points in the space by open sets (and closed sets).

\begin{thm}[Topological restatement of the Angluin Theorem]\label{thm:introtopAngluin}
Suppose the set of underlying strings is countable.
    In the Gold-Angluin full enumeration model, the identification in the limit is possible if and only if (1) the set of languages is countable, and (2) in $\mathcal{T}_{\rm{full}}$
    every point $L$ satisfies \( \overline{\{L\}} \setminus \{L\} \) is closed. In other words, this is $T_D$ space.
\end{thm}
Note that a
$T_D$ space is a standard notion in the hierarchy of topological spaces, capturing one of the intermediate levels of point separation by open sets.

When we move next to characterize the language collections
that admit conjunction-based identification in the limit 
with partial enumeration, the situation becomes significantly
more intricate. Nevertheless, we are able to provide a surprisingly
clean characterization of this property using a topological definition.
The characterization is again based on the $T_D$ separation property,
but now it is applied to a topological space
$\tau_C$ that is defined in detail in Subsection~\ref{subsec:toppartial};
essentially, it is like the space $\mathcal{T}_{\rm{full}}$ except that
the sets in its basis also include a constraint on the cardinality
of intersections with $C$.

\begin{thm}\label{thm:introidentificationpartial}
    Identification in the limit is possible in the partial enumeration model if and only if for any language $L$ and any infinite $C \subset L$, the Kolmogorov quotient of the space $\topo{C}$ is a $T_D$ space.
\end{thm}

This characterization is quite robust with respect to exactly how the algorithm chooses a representation for its languages $M_t$: when the required $T_D$ condition holds the algorithm can choose $M_t$ to be a finite intersection of languages in $\mathcal{X}$, as noted above; and when the $T_D$ condition does not hold, identification in the limit is not possible in the partial enumeration {\em even if} the algorithm is allowed to output arbitrary languages (disregarding the representational issues raised by an algorithm that is able to output arbitrary infinite languages).
Also, as noted above, it is straightforward to check that very little is possible in the partial enumeration model if the algorithm must output individual languages from $\mathcal{X}$, and so allowing finite intersections is arguably the natural baseline.

From the topological characterization, we are also able to establish an important additional robustness property:  identification in the limit with partial enumeration is not affected by the deletion of a finite subset of strings from the underlying ground set $\groundset$.  That is, if $W$ is any finite set of strings in the ground set $\groundset$, then a collection of languages $\mathcal{X}$ is identifiable in the limit with partial enumeration over the ground set $\groundset$ if and only if it is identifiable in the limit with partial enumeration over the ground set $\groundset - W$.
(We naturally require that any two languages $L_i, L_j \in \mathcal{X}$ that become the same after the deletion of $F$ are treated as the same language for purposes of identification.)
We can state this result as follows:
\begin{thm}
\label{thm:introRobustness}
    Let $\mathcal{X}$ be a countable collection of languages, and let $\mathcal{X}'$ be obtained from $\mathcal{X}$ by removing a finite set of strings from the underlying ground set $\groundset$ (thus removing those strings from each language and eliminating any duplicate languages that may result). Then $\mathcal{X}$ is identifiable in the limit with partial enumeration if and only if $\mathcal{X}'$ is identifiable in the limit with partial enumeration. 
\end{thm}

It is interesting to have this robustness property in the partial enumeration model, since it does not hold in the classical full enumeration model of Gold and Angluin; there, the property is ``brittle'' in the sense that there are instances where deleting even a single string from the ground set $\groundset$ can change a collection $\mathcal{X}$ that is identifiable in the limit to one that is not identifiable in the limit.

Using our topological definitions, we also examine additional separation properties, revealing their close connections to the spaces we define,
and showing how this topological perspective naturally guides the search for suitable algorithms.

Beyond providing a topological proof of Angluin’s theorem, the special case of Theorem~\ref{thm:introidentificationpartial} applied to the full enumeration model can also be viewed as an \emph{algorithmic} proof of Angluin’s result: one that does not require explicit knowledge of the certificate (tell-tale set) for each language.

\section{Generation in the limit with an algorithm that is accurate infinitely often}\label{sec:generating}

In order to achieve the positive density guarantees for the partial enumeration model, we begin by proving
Theorems~\ref{thm:introKWnew} and~\ref{thm:introaccuate}. We first show there is an algorithm that can achieve generation in the limit in the presence of an adversary that can perform partial enumeration.
As noted in the introduction, this cannot be achieved by an algorithm that must output an index $i_t$ in every step $t$ such that $L_{i_t} \subseteq K$ after some finite time, but we will see from the of Theorem~\ref{thm:introKWnew} that it can be achieved by an algorithm that produces a finite intersection $M_t$ of languages in $\mathcal{X}$.  

This use of finite intersections as representations, what we call {\em conjunction-based generation}, or {\em semi-index-based generation}, will be the basis of our stronger algorithms as well.
As the first step in this stronger direction, we say that a conjunction-based algorithm is 
{\em accurate} in step $t$ if the finite intersection $M_t$ that it produces in step $t$ contains the adversary set $C$: that is, $C \subseteq M_t \subseteq K$.  In the latter portion of this section, we will prove Theorem \ref{thm:introaccuate}, that there is a conjunction-based generation algorithm that is accurate infinitely often.

We first prove the following weaker statement.

\begin{thm}[Restatement of Theorem \ref{thm:introKWnew}]
 Let $\mathcal{X}$ be a countable collection of languages. Let $K$ be the true language. A learner observing any infinite subset \(C \subseteq K\) in any enumeration can successfully generate strings from the true language \(K\) in the limit.
\end{thm}

\begin{proof}
At time~$t$, we say that a language is \emph{consistent} if it contains all strings revealed by the adversary up to that point.

After each time stamp, we eliminate all languages that are not consistent, while preserving their relative order in the original language ordering. Suppose that at time~$t$, the remaining consistent languages, in order, are 
\[
L_{i_1}, L_{i_2}, \dots \quad \text{where } 1 \le i_1 < i_2 < \dots.
\]
For each~$j$, consider the intersection 
\[
L_{i_1} \cap L_{i_2} \cap \dots \cap L_{i_j}.
\]
The idea is to find the largest~$j$ for which this intersection is infinite; denote this index by~$j^*$. Note that $j^* \ge 1$. Specifically, we have the following cases.

\begin{enumerate}
    \item {Case 1:} If there exists $j^*$ such that $L_{i_1} \cap \dots \cap L_{i_{j^*}}$ is infinite, but $L_{i_1} \cap \dots \cap L_{i_{j^* + 1}}$ is finite, then output any string in $L_{i_1} \cap \dots \cap L_{i_{j^*}}$.

    \item {Case 2:} If the intersection of all the consistent sets, $L_{i_1} \cap L_{i_2} \cap \dots$, is infinite, then output any string in this intersection.

    \item {Case 3:} If the intersection of all the consistent sets is finite (including being empty), but for every positive integer $m \geq 1$, $L_{i_1} \cap \dots \cap L_{i_m}$ is infinite, then if the current timestamp is $t$, output any string in $L_{i_1} \cap \dots \cap L_{i_t}$.
\end{enumerate}

To show that this algorithm generates in the limit, we will prove that after some finite time, all remaining consistent languages preceding~$K$ (with respect to the original language ordering) will contain~$C$. Indeed, if some language~$L$ does not contain an element~$c \in C$, then once~$c$ is revealed by the adversary,~$L$ will be eliminated. Since there are only finitely many languages preceding~$K$ in the original language ordering, it follows that after some finite time~$t > z$, where~$K$ is the $z$-th language in the original ordering, we have~$i_{j^*} \ge z$. Consequently, from that point onward, every output string will belong to~$K$. Our condition requiring that the specific intersection have infinite cardinality guarantees that there is always at least one string available for the algorithm to output at any stage of the learning process.
\end{proof}

\subsection{Algorithm which sees $C$ infinitely often}\label{subsec:accurate}

In this subsection we will prove Theorem \ref{thm:introaccuate}, restated as below.
\begin{thm}\label{thm:introaccuate}
There exists an algorithm that generates in the limit, such that 
\begin{enumerate}
    \item After some finite time, the set of strings the algorithm identifies at time $t$ is an infinite subset of $K$; and 
    \item For infinitely many time steps $t$, the set of strings the algorithm identifies at time $t$ fully contains the set $C$.
\end{enumerate}
\end{thm}

Note that this theorem is substantially stronger than Theorem~\ref{thm:introKWnew}, particularly in its ``breadth'' guarantee.  
In Theorem~\ref{thm:introKWnew}, if we remain indefinitely in Case~3, one can imagine the set of candidate languages considered by the algorithm becoming progressively thinner: since the learner cannot determine which of the first $t$ hypotheses is correct, its current conjecture is formed by intersecting an ever-increasing number of languages.  
As $t$ grows, this intersection becomes sparser within both the true language and the context set $C$, reflecting the learner’s cautious attempt to avoid overshooting after some finite time.  
At first glance, this outcome may appear unavoidable under such a constraint.  
However, Theorem~\ref{thm:introaccuate} shows that this pathological situation cannot persist indefinitely.  
Infinitely often, the learner will ``zoom in'' on an intersection that fully contains $C$, thereby restoring full breadth and ensuring that the learning process repeatedly recovers hypotheses that encompass the true context.  

Nevertheless, this theorem by itself does not yield a tight bound on the breadth, since it provides no guarantee on the \emph{gaps} between these ``infinitely often'' accurate stages: the accurate moments could, in principle, occur at arbitrarily sparse timestamps.  
This theorem will, however, play an important role later, when we establish a tight quantitative bound on the breadth of generation.

We begin with some observation.

\medskip
\noindent\textbf{Semi-index based generation.} 
We say that an algorithm is a \emph{semi-index based generation} if, at each time step~$t$, it outputs a finite set of indices corresponding to languages whose intersection is infinite. We say that an algorithm is a \emph{semi-index based generation in the limit} if there exists a time~$T$ such that for all~$t \ge T$, the intersection of the languages corresponding to the indices output by the algorithm at time~$t$ is contained in the true language~$K$.

\begin{lem}\label{lem:semiindex}
Any deterministic algorithm that is element-based generation in the limit can be translated into a deterministic algorithm that is a semi-index based generation, and vice versa. Furthermore, the optimal output density achievable by the best element-based generation algorithm in the worst case coincides with that achievable by the best semi-index based generation algorithm in the worst case.
\end{lem}
\begin{proof}
Suppose that~$A$ is an element-based generation algorithm in the limit. We construct a corresponding semi-index based generation algorithm~$A'$. At time~$t$, let the output of~$A$ be denoted by~$o_t$. Then~$A'$ outputs at most~$t$ indices corresponding to the languages with the smallest indices satisfying the following conditions:
\begin{enumerate}
    \item each language is consistent with all adversary outputs $w_1, w_2, \dots, w_t$ up to time~$t$;
    \item each language contains the element~$o_t$ generated by~$A$ at time~$t$; and
    \item the intersection of these languages is infinite.
\end{enumerate}
If no such language exists, then~$A'$ outputs the index of the smallest language that contains~$o_t$.

Clearly,~$A'$ is a semi-index based generation algorithm. It remains to show that~$A'$ also generates in the limit. To see this, consider all languages whose indices are smaller than that of~$K$. For each such language that does not fully contain the adversary enumeration~$C$, it will become inconsistent after some finite time. Since there are only finitely many languages with indices smaller than that of~$K$, there exists a finite time~$T$ after which every consistent language with index smaller than that of~$K$ fully contains~$C$.  

Since~$A$ generates in the limit, we may take~$T$ to be the larger of the two time thresholds—one ensuring that~$A$ always outputs elements from~$K$ and the other ensuring consistency as described above. Therefore, for all~$t \ge T$, the language~$K$ satisfies both (1) it is consistent at time~$t$, and (2) it contains the element~$o_t$. Hence, as long as~$T$ exceeds the index of~$K$, the index of~$K$ will always be included among those output by~$A'$.  

Furthermore, after time~$T$, the intersection of all consistent languages with indices smaller than that of~$K$ remains infinite, since all these languages fully contain~$C$. Consequently, the indices output by~$A'$ necessarily include the index of~$K$. Because~$A'$ outputs multiple indices and we only consider the intersection of the corresponding languages, this intersection is always a subset of~$K$. Thus,~$A'$ also generates in the limit.

    The other direction is straightforward. For any algorithm~$A'$ that is a semi-index based generation in the limit, define an element-based generation algorithm~$A$ as follows: at each time~$t$, let~$A$ output the smallest unused string that lies in the intersection of the languages whose indices are output by~$A'$ at time~$t$.  

By this construction, at each time~$t$, the element generated by~$A$ is no smaller than the element that could be generated by~$A'$. Moreover, if~$A$ outputs a string from~$K$, then so does~$A'$. Hence, the set of outputs~$O(A)$ is strictly ``behind'' the set of outputs~$O(A')$ within~$K$. Consequently, for any family of languages and any adversary enumeration~$C \subseteq K$, both the upper and lower densities of~$O(A')$ in~$K$ are at least those of~$O(A)$.  

Since every semi-index based generation algorithm can be viewed as an element-based generation algorithm, it follows that the optimal achievable output density in the worst case is the same for both models.
\end{proof}

Note that the above equivalence works for {\it any} algorithm that generates in the limit, even in  the case where the adversary might enumerate the full set $K$.

% \begin{thm}
%     In the partial enumeration KM model, suppose the adversary enumerate lower density of $c$ fraction of the true language $K$, there is an algorithm that generates in the limit which could generate a lower density of at least $0.1c$ of elements in $K$. 
% \end{thm}

By Lemma~\ref{lem:semiindex}, it suffices to consider only the semi-index–based generation model.  
Note that in Lemma~\ref{lem:semiindex}, the number of indices produced by the algorithm can be as large as $t$ at each timestamp $t$, so the size of the semi-index may grow rapidly with time.  
This, however, serves primarily as a \emph{conceptual} argument demonstrating that the semi-index model is sufficient to capture the full generality of the setting.  
In practice, our actual algorithm will be far more efficient: in many cases, the number of indices it maintains will remain bounded, rather than growing unboundedly with~$t$.

Suppose that at time~$t$, the intersection of the finitely many languages whose indices are chosen by the algorithm is denoted by~$\I{t}$.

% The theorem below is a generalization to the previous theorem where when $C = K$, the algorithm is accurate infinitely often. 
It suffices to prove the following theorem to show that there is an algorithm that is accurate infinitely often. 
\begin{thm}\label{thm:accurateoften}
There exists a semi-index based generation algorithm that generates in the limit, such that for infinitely many time steps $t$, the intersection $\I{t}$ of the languages whose indices are output by the algorithm at time $t$ contains the set $C$.
\end{thm}
\begin{proof}
Now we describe the algorithm to determine $\I{t}$.

At time~$t$, we consider only the languages that remain consistent up to time~$t$, i.e., those for which all adversary enumerations $w_1, \dots, w_t$ are contained in the language.  
List these languages from left to right according to their original order in the list of languages, and denote them by
$
\Lc{1}{t}, \Lc{2}{t}, \dots.
$
Define
\[
\Ic{1}{t} = \Lc{1}{t}, \qquad 
\Ic{2}{t} = \Lc{1}{t} \cap \Lc{2}{t},
\]
and, in general, for $i \ge 1$,
\begin{equation}
    \Ic{i}{t} = \bigcap_{j=1}^i \Lc{j}{t}.
\end{equation}

In the case where only finitely many consistent languages remain at time~$t$—say, exactly~$k$ of them—for notational convenience we set
\[
\Lc{i}{t} = \Lc{k}{t} \quad \text{for all } i \ge k.
\]
In other words, the sequence of consistent languages stabilizes at~$\Lc{k}{t}$ and remains constant thereafter at time~$t$.

Clearly, at each time $t$, the sequence $\Ic{i}{t}$ forms a {\it descending chain of intersections}, i.e., 
\begin{equation}
  \mCh{t} =   \Ic{1}{t} \supseteq
 \Ic{2}{t} \supseteq
 \Ic{3}{t} \supseteq
 \dots. 
\end{equation}

\noindent\textbf{Algorithm 1. Finding {identified intersection} $\I{t}$.}
Now we will define the {\it identified intersection} $\I{t}$ for each $t$. 
When $t=1$, let $\I{t} = \Ic{1}{1}$. 

Suppose we already have the identified intersection $\I{t}$ for time $t$. For time $t+1$, there are several scenarios.

\begin{enumerate}
    \item If the descending chain of intersections at time $t$ is the same as time $t+1$, and suppose $\I{t} = \Ic{k}{t}$ for some finite $k$. Let $\I{t+1} = \Ic{k+1}{t} = \Ic{k+1}{t+1}$ if $\Ic{k+1}{t}$ is an infinite set, and let $\I{t+1} = \Ic{k}{t} = \Ic{k}{t+1}$ if $\Ic{k+1}{t}$ is a finite set. 
    \item If the descending chain of intersections at time $t+1$ is different from time $t+1$: let $k^*$ be the largest integer such that $\Ic{k^*}{t} = \Ic{k^*}{t+1}$ and $\Ic{k^*}{t}$ is infinite (which $k^*$ could be zero for example if $\Ic{1}{t} \neq \Ic{1}{t+1}$.)  Let $\I{t+1} = \Ic{k^*}{t}$ if such a $k^* \geq 1$; and $\I{t+1} =\Ic{1}{t+1} $ if $k^* = 0$. 
\end{enumerate} 

\textbf{Properties of $\I{t}$. }
Suppose an adversary plans to enumerate an infinite set~$C$ within the true language~$K$.  
We say that the intersection~$\I{t}$ is \emph{valid} if $\I{t} \subset K$.  
This condition means that every potential string in~$\I{t}$ belongs to~$K$, and thus outputting any string from~$\I{t}$ is valid.  
We say that~$\I{t}$ is \emph{full} if $C \subset \I{t}$, that is, if the entire adversary enumeration~$C$ is contained in~$\I{t}$.

Throughout the proof, all finite constants are understood to depend on~$C$,~$K$, and the underlying list of languages, unless specified otherwise.

\begin{lem}\label{lem:property} The following statements hold.

    \begin{enumerate}
        \item There is a finite time $T$  such that after which all the identified intersection $\I{t}$ are valid.
            \item After some finite time $T$, either $\I{t} \subset \I{t+1}$ or $\I{t+1} \subset \I{t}$. 

        \item The identified intersection $\I{t}$ will be full for infinitely many $t$. 
    \end{enumerate}
\end{lem}

\begin{proof}

We first prove Item~1. Observe that the true language~$K$ is always consistent.  
We claim that there exists a finite time~$T$ such that every consistent language whose position precedes that of~$K$ contains~$C$.  
This is the same as earlier arguments: note that for any language~$L$ that does not fully contain~$C$, since the adversary eventually enumerates every element of~$C$, there must exist a time at which a string in~$C \setminus L$ is revealed by the adversary.  
At that point,~$L$ becomes inconsistent.  
Because there are only finitely many languages whose positions in the original ordering come before~$K$, all such languages will become inconsistent after some finite time.  
Hence, the claim follows.

    Also notice that for each language $L$, if it fully contains $C$, then it will always be consistent. 

    So we can assume there is a time $T$ and an integer $k \geq 1$, such that for all $t \geq T$, 
    \begin{equation}
        C \subset \Lc{1}{t} =   \Lc{1}{t+1},\ \  C \subset \Lc{2}{t} =   \Lc{2}{t+1}, \dots,\ 
    \ C \subset \Lc{k}{t} =   \Lc{k}{t+1} =K.
    \end{equation}

    In particular, this means that for all $t \geq T$ and $i \leq k$, 
    \begin{equation}
        \Ic{i}{t} = \Ic{i}{t+1} = \Ic{i}{t+2} =\dots
    \end{equation}
    and \begin{equation}
        C \subset \Ic{k}{t} \subset K.
    \end{equation}

    By our choice of $\I{t}$, we have that for $t > T$, $\I{t} \subset \Ic{k}{t-1} \subset K$. Therefore $\I{t}$ is always valid for $t > T$.

We now prove Item~2.  
By Scenario~1 above, when the descending chain satisfies $\mCh{t} = \mCh{t+1}$, the intersection sets chosen at times~$t$ and~$t+1$ lie on the same chain and therefore will form an inclusion relation.  
Similarly, in Item~2 above, unless $k^* = 0$, both $\I{t}$ and $\I{t+1}$ lie on the chain~$\mCh{t}$ and thus must form an inclusion relation.  
It therefore suffices to show that after some finite time~$T$, whenever we are in Scenario~2, we have $k^* \ge 1$.  
From the argument above, we already know that $k \ge 1$, and since $k^* \ge k$, it follows that $k^* \ge 1$.

We are left to prove Item~3.  
We will show that there exist infinitely many times~$t$ such that $C \subset \I{t}$.  

Suppose that at the finite time~$T$ identified above, all consistent languages $\Ic{i}{T}$ for $i \ge k$ (and hence for all $i \ge 1$) contain~$C$.  
Then all these languages remain consistent for every~$t \ge T$, and the sequence of consistent languages remains unchanged for all such~$t$.  
Since~$C$ is infinite, it follows that $\I{t}$ always contains~$C$, as desired.  

Now suppose instead that at some time~$T_1 \ge T$, one of the previously consistent languages becomes inconsistent.  
Let~$k'$ be the smallest index such that $\Lc{k'}{T_1}$ does not fully contain~$C$.  
By our choice of~$T$, we have $k' > k$.

Since $\Lc{k'}{T_1}$ does not fully contain~$C$, there must exist a first time~$T_2 > T_1$ at which the adversary outputs a string in $C \setminus \Lc{k'}{T_1}$.  
When this occurs, the language $\Lc{k'}{T_1} = \Lc{k'}{T_2 - 1}$ becomes inconsistent at time~$T_2$.  

At time~$T_2$, we have $\Lc{i}{T_2 - 1} = \Lc{i}{T_2}$ for all $i \le k' - 1$ since they all contain $C$, but $\Lc{k'}{T_2} \ne \Lc{k'}{T_2 - 1}$.  
By Rule~2, it follows that
\[
\I{T_2} = \Ic{k' - 1}{T_2 - 1} = \bigcap_{i \le k' - 1} \Lc{i}{T_2 - 1}.
\]
By our choice of~$k'$ and~$T_2$, we have $\Lc{i}{t} \supset C$ for all $i \le k' - 1$ and all $t \ge T_1$.  
This implies that $\I{T_2}$ is infinite and contains~$C$.  

We can repeat this argument each time the leftmost consistent language fails to fully contain~$C$.  
If this never happens again, it simply means that all remaining consistent languages fully contain~$C$.

Hence, the identified intersection set $\I{t}$ will be full infinitely often.
\end{proof}
\end{proof}

\section{Achieving Positive Lower Density in the Partial Enumeration Model}\label{sec:tension}

We now return to the first of our main questions in this paper, to establish what can be said about {\em breadth} in language generation when the adversary may only enumerate a subset $C$ of the true language $K$.
As in the introduction, we formulate this question in terms of {\em density}: the algorithm is achieving breadth if the set of all strings it outputs has sufficiently large lower density in $K$.
With this in mind, let $O(E,\ma)$ as before denote the set of all strings output by the algorithm over all time steps, when presented with an enumeration $E$ of an infinite subset $C$ of $K$.

The central question, originally posed in  in~\cite{kleinberg2025density} for the special case of full enumeration, asks how high this density can be without compromising validity.  
This question captures the fundamental tension between \emph{breadth} and \emph{validity}: generating too broadly risks producing strings outside $K$ (analogous to hallucination), while generating too conservatively leads to sparse or mode-collapsed outputs.  
Understanding the optimal trade-off between these forces is both conceptually important and technically delicate, as it reflects the theoretical limits of language generation in the limit.  

In this section we work in the more general partial enumeration model, where the set $C$ enumerated by the adversary itself only has some lower density $d(C,K)$ in the full true language $K$.
Our goal in this section is to determine the precise boundary of the trade-off between breadth and validity by proving that the constant $d(C, K)/2$ represents a sharp, unimprovable limit on achievable breadth under worst-case adversarial enumeration. In other words, 
if the adversary outputs a set $C \subseteq K$ that has lower density $\alpha > 0$, then the algorithm's output set $O(E,\ma)$ has lower density at least $\alpha/2$.

As a special case, this theorem resolves the open question in  \cite{kleinberg2025density} (and similar questions in \cite{kleinberg2024limit}), about the tight lower bound for output lower density in the standard (full) enumeration model of language generation in the limit. 
Our proof in this section will make crucial use of Algorithm~1 from Subsection~\ref{subsec:accurate}.

\subsection{Warm-up: A Weaker Bound}
We first establish a weaker bound, accompanied by a simpler proof.  
In the next subsection, we build upon this argument to obtain the optimal bound.  
This proof also provides additional insights compared to~\cite{kleinberg2025density}, yielding a more natural and streamlined argument with an improved bound.

\begin{thm}\label{thm:ldweak}
There exists an algorithm~$\ma$ that achieves element-based generation in the limit in the partial enumeration model and satisfies the following property.  
Given any countable collection of languages~$\coll$ with an underlying ordering of the strings in each language, and for any adversarial partial enumeration~$C$ of some language $\trueL \in \coll$ such that the adversary enumerates at least a lower density~$\alpha > 0$ fraction of the strings in~$\trueL$,  
the set of output strings~$O(C, \ma)$ generated by the algorithm has lower density in~$K$ at least~$\alpha/3$.
\end{thm}

Note that in~\cite{kleinberg2025density}, within the full enumeration model, a lower density bound of at least~$1/8$ was established.  
This corresponds to a special case of the present theorem with~$\alpha = 1$—in fact, it is an \textbf{even stricter} case in \cite{kleinberg2025density}, since the full enumeration model does not allow the adversary to omit even a set of density zero, whereas such omissions are permitted when~$\alpha = 1$ in our setting.  
Thus, our result not only generalizes but also strengthens the earlier bound of~$1/8$.  
Moreover, the proof presented here is considerably simpler than the more intricate argument in~\cite{kleinberg2025density}, thereby improving readability.

\vspace{0.1in}
\noindent\textbf{Aggressive set, priority string list, and token.}
Recall $\I{t}$ is the language identified by the algorithm $\Acc$ at time $t$. 

At each time $t$, the language will decide on 
\begin{enumerate}
\item an \emph{aggressive set}~$\fbga{t}$, which consists of strings that the language guess believes could be part of the true language, though including them may risk overshooting;
    \item a \emph{priority string list} $\mS_t$, which contains the strings that have the priority to be output by the algorithm.
\end{enumerate}
In the algorithm description later, we will explain how we chose $\fbga{t}$ and $\mS_t$.

By saying that the algorithm {\it aggressively guesses} $\fbga{t}$ at time $t$ with {\it token} $N_t$, we mean the following:

Suppose that at time~$t$, the string~$w$ is the largest string (with respect to the universal ordering of strings) that has been produced by either the input or the adversary up to time~$t$.  
Let $w' \in \fbga{t}$ be the smallest string in~$\fbga{t}$ that is larger than~$w$.  

We then add to the set~$\mS_{t-1}$ all unused strings in~$\fbga{t}$ that are at most~$w'$.  
In addition, we include the next~$N_t$ smallest strings in~$\fbga{t}$ that are larger than~$w'$ into~$\mS_{t-1}$.  
(Note that these $N_t$ strings have not been used by either the algorithm or the adversary.)  

Next, we order the strings in the updated set~$\mS_{t-1}$ from smallest to largest, and output the first unused string in~$\mS_{t-1}$.  
Finally, we update~$\mS_t$ by removing the most recently output string from~$\mS_{t-1}$.

\vspace{0.1in}
\noindent\textbf{Algorithm.}
\begin{enumerate}
\item Initialize the fallback list set $\mS_0 = \{\emptyset \}$. Set token $N_0 = 0$.
\item If $\I{t+1} = \I{t}$, do not fall back to any language. If $\mS_t \neq \emptyset$ and the smallest unused string in $\mS_t$ is smaller than the next available string in $\I{t}$, output the smallest unused string in $\mS_t$, and update $\mS_t$ to be $\mS_{t+1}$ by removing this string and all the strings already used by the adversary. Otherwise, output the smallest unused string in $\I{t}$ and set $\mS_{t+1}$ to be $ \mS_t$ after removal the strings already generated by the adversary. 
\item If $\I{t+1}$ is a strict subset of~$\I{t}$, we make an aggressive guess by setting $\fbga{t+1} = \I{t}$.  
Note that $\I{t+1}$ is a strict subset of~$\I{t}$ only when both sets belong to the descending chain of intersections at time~$t+1$, namely,
\[
\mCh{t+1} = \lanj{1} \supsetneq \lanj{2} \supsetneq \lanj{3} \supsetneq \dots.
\]
Suppose $\I{t}$ corresponds to the set~$\lanj{i}$; then, by Algorithm~1, we have $\I{t+1} = \lanj{i+1}$.
    % Set the token $N_{t+1} = 2(j-i)$. 

\item If $\I{t+1}$ is a strict superset of $\I{t}$. In this case we aggressively guess $\I{t+1}$, i.e., set $\fbga{t+1} = \I{t+1}$.  Set the token $N_{t+1} = 2$. 

\item \label{enum:alginf5} 
Else, we look for the minimal (with respect to set inclusion) identified intersection~$I$ such that~$I$ is a superset of (i.e., precedes)~$\I{t}$ in~$\mCh{t}$ and also a superset of~$\I{t+1}$ in~$\mCh{t+1}$, if such an~$I$ exists.  
Note that if there exist identified intersections~$I$ satisfying these conditions, then the minimal one (under set inclusion) must exist, since only finitely many sets precede~$\I{t}$ in~$\mCh{t}$ and finitely many sets precede~$\I{t+1}$ in~$\mCh{t+1}$.

    \begin{enumerate}
\item 
If no such~$I$ exists, then no fallback is performed.  
If~$\mS_t \neq \emptyset$ and the smallest unused string in~$\mS_t$ is smaller than the next available string in~$\I{t}$,  
output the smallest unused string in~$\mS_t$, and update~$\mS_t$ by removing this string together with any strings used by the adversary;  
denote the resulting set by~$\mS_{t+1}$.  
Otherwise, output the smallest unused string in~$\I{t}$ and set~$\mS_{t+1} = \mS_t$ after removing the strings that have been used by the adversary.
        \item
If such an $I$ exists, then aggressively guess $\fbga{t+1} = I$. 

Since  $I = \fbga{t+1}$ is also part of $\mCh{t+1}$. Suppose $\mCh{t+1} = \lanj{1} \supsetneq \lanj{2} \supsetneq \dots$. 
Suppose $\fbga{t+1} = \lanj{i}$ and $\I{t+1} = \lanj{j}$. 
Set the token $N_{t+1} = 2(j-i)$. 
    \end{enumerate}
\end{enumerate}

We now prove the validity of the algorithm.
\begin{lem}\label{lem:validityinfinite}
There exists a finite time~$T^*$, possibly depending on the adversary’s enumeration of strings in~$K$, such that for all~$t \ge T^*$, the algorithm always outputs a string in~$K$.
\end{lem}
\begin{proof}
Since each time we make an aggressive guess, we add only finitely many strings from the corresponding aggressive set to the priority list.  
By Lemma~\ref{lem:property}, Item~2, after some finite time~$T$, we have either $\I{t} \subset \I{t+1}$ or $\I{t+1} \subset \I{t}$.  
Thus, from that point onward, only Items~2--4 of the algorithm above apply.  
In each of these cases, the aggressive set is always either~$\I{t}$ or~$\I{t+1}$.  
The claim then follows from Lemma~\ref{lem:property}, Item~1.
\end{proof}

Again, at time~$t$, let~$w_t$ denote the string revealed by the adversary, and let~$o_t$ denote the string output by the algorithm at time~$t$.  
Define
\[
O = \bigcup_{t=1}^\infty \{o_t\}
\]
as the set of all strings output by the algorithm.  
Clearly, $O \subsetneq K$, since the algorithm never repeats a string already generated by the adversary.

Without loss of generality, we may assume that the ground set is the set of natural numbers~$\mathbb{N}$.  
Then~$K$ is a subset of~$\mathbb{N}$, consisting of the ordered integers
\[
K = \{\psi(1), \psi(2), \dots\},
\]
where the $i$-th string in~$K$ is~$\psi(i)$.  
For each string~$x \in K$, recall that $\Succ_K(x)$ denotes the successor of~$x$ in~$K$, that is,
\[
\Succ_K(x) = \psi(\psi^{-1}(x) + 1).
\]

\

We first prove a much easier version of Theorem \ref{thm:ldweak}, i.e., instead of lower density, we show that the {\it upper density} is arbitrarily close to $\alpha/2$.

\begin{thm}\label{thm:halfupperbound}
    There exists an algorithm that guarantees validity in the partial enumeration model, where the adversary enumerates a lower density at least $\alpha > 0$ in $K$. Then the output $O$ has {\it upper density} at least $\alpha/2$ in the true language $K$. Furthermore, this bound is sharp. 
\end{thm}

\begin{proof}
We first see that $\alpha/2$ is sharp. 

We first prove that when~$\alpha = 1$, the upper density~$1/2$ is indeed the best upper bound benchmark one can hope for.  
This argument is essentially the same as that in~\cite{kleinberg2025density}.  

Consider, for instance, the case where there is only one language—namely, the true language~$K = \mathbb{N}$.  
At any time~$t$, the algorithm cannot output a string that has already been produced by the adversary.  
Suppose the adversary proceeds as follows:  
at each time~$t$ that is \emph{not} a power of~$2$, it outputs the smallest string that has not yet been used by either the adversary or the algorithm;  
and at each time~$2^n$, it outputs the smallest integer that has previously been output by the algorithm.  
It is easy to verify that the adversary eventually enumerates all strings.  

We now show that the algorithm can produce outputs with upper density at most~$1/2$.  
To see this, consider the first~$n$ integers.  
Suppose~$o_1$ is output by the algorithm at time~$t$.  
Then, according to the adversary’s strategy, at time~$t+1$, unless~$t+1$ is a power of two, the adversary will output the smallest unused string—which must be less than or equal to~$o_1 + 1$.  
Therefore, the number of outputs from the algorithm among the first~$n$ steps is bounded by~$n/2 + O(\log n)$.  
Consequently, the algorithm’s upper density is at most~$1/2$.

We now prove that~$\alpha/2$ is the best benchmark attainable for any~$\alpha \in [0,1]$.  
Indeed, consider the lower bound construction consisting of a collection of languages~$\mX$ and a true language~$K$, for which the full enumeration model cannot achieve an upper density exceeding~$1/2$.  
Now introduce an additional language~$K'$ satisfying~$K \subset K'$, where~$K$ has lower density at least~$\alpha$ within~$K'$, and assume that the new elements of~$K'$ (those in~$K' \setminus K$) are strings never used in any languages in $\mX$.  

Treat~$K'$ as the new true language, and let the adversary enumerate only elements of~$K$, which satisfies the partial enumeration model.  
If the algorithm guarantees validity, it should not be able to distinguish whether the true language is~$K$ or~$K'$.  
After some finite time~$T$, the algorithm must therefore output only strings from~$K$.  
By the same reasoning as in the previous argument, the upper density of~$O$ within~$K$ cannot exceed~$1/2$.  
Hence, the upper density of~$O$ within~$K'$ cannot exceed~$\alpha/2$.

\

We now show that the upper density can indeed be achieved as~$\alpha/2$.  
Let~$T_i$ denote the time stamps at which the algorithm is \emph{full}.  
We may restrict attention to the time stamps following the point at which the algorithm becomes \emph{valid}.  
Let~$a_i$ be the largest string (with respect to the universal ordering) that has been produced by either the adversary or the algorithm by time~$T_i$.  

At time~$T_i$, observe first that the number of strings output by the algorithm equals the number of strings output by the adversary.  
After time~$T_i$, since all strings in~$\I{T_i}$ smaller than~$a_i$ are placed in the priority list~$\mS_{T_i}$ and since~$C \subset \I{T_i}$, these strings have priority for output.  
That is, before all strings in~$\I{T_i}$ smaller than~$a_i$ are output, the algorithm cannot output any string in~$K$ that is larger than~$a_i$.  

Hence, for the remaining strings in~$K$ smaller than~$a_i$, before the set~$\I{T_i} \cap [a_i]$ is exhausted, the algorithm either outputs the smallest unused string in~$\I{T_i} \cap [a_i]$ or a string in~$K$ smaller than the smallest unused string in~$\I{T_i} \cap [a_i]$.  
Therefore, within~$K \cap [a_i]$, the algorithm outputs at least
\[
\frac{|\I{T_i} \cap [a_i]|}{2} - o(a_i)
\]
strings belonging to~$K \cap [a_i]$.  

Since~$C$ has lower density at least~$\alpha$ in~$K$, and since~$C \subset \I{T_i}$ for varying~$i$, it follows that each~$\I{T_i}$ has lower density at least~$\alpha$ in~$K$.  
In particular, as~$T_i \to \infty$, the set~$\I{T_i} \cap [a_i]$ has density~$\alpha - o(1)$ in~$K$.  
Consequently, the quantity
$\frac{|\I{T_i} \cap [a_i]|}{2} - o(a_i)$
has density~$\alpha/2 - o(1)$ in~$[a_i]$ as~$i \to \infty$.  
Therefore, the algorithm’s output has upper density at least~$\alpha/2$ in~$K$.
\end{proof}

The reason that Theorem~\ref{thm:halfupperbound}, which establishes the sharp upper density bound of~$\alpha/2$, cannot be directly adapted to yield a corresponding lower density bound of~$\alpha/2$ is rooted in the behavior described earlier in Theorem~\ref{thm:introaccuate}.  
There, the algorithm achieves \emph{full} breadth at infinitely many timestamps $T_i$, but there is no guarantee on the \emph{gaps} or spacing between successive $T_i$’s.  
Although the algorithm can artificially delay its output at any stage, for example, by adding many more strings to the priority list~$\mS_{T_i}$, it has no knowledge of when the next $T_{i+1}$ will occur.  
Importantly, the existence of a subsequent ``full'' timestamp is \emph{guaranteed}, yet the algorithm will never be aware of when it actually happens: there is no signal, trigger, or observable event that marks its arrival.  

The above analysis does not rule out the possibility that an adversarial enumeration could, in principle, exploit this epistemic limitation by releasing a large number of strings from~$K$ between two consecutive full timestamps.  
Such an adversary can thereby create extensive regions (``deserts'') of~$K$ where the algorithm’s outputs are only sparsely present.  
As a result, while the upper density can still be tightly bounded by~$\alpha/2$, the lack of control over the temporal distribution of full stages prevents an analogous guarantee for the lower density.

We are now ready to prove Theorem \ref{thm:ldweak}, where new ideas of the analysis are needed.
\begin{proof}[Proof of Theorem \ref{thm:ldweak} given the algorithm]
By the preceding arguments, it only remains to show that the lower density of~$O$ in~$K$ is at least~$\alpha/3$.

Let $\mB_g$ be the set of ``good" strings in $C \setminus O$, defined as \[\mB_g = \{ w_t \in C \setminus O: o_{t-1} \leq \Succ_K(w_t) \}.\]
This set consists of adversary inputs for which the algorithm's previous output is not far behind. 

Let $\mB_b$ be the set of ``bad" strings, defined as \[\mB_b = C \setminus O \setminus \mB_g.\] 

The next lemma is the main lemma in the proof. 

\begin{lem}\label{claim:infrho}
There exists a finite time $T$ and a map $\rho: \{w_t: w_t \in \mB_b, t > T\} \to O$ such that $\rho(w_t) < w_t$ in the string ordering of $K$. In addition, $\rho$ is injective. 
\end{lem}  
\begin{proof}
By Lemma~\ref{lem:property}, we may assume that we are beyond the finite time~$T_1$ after which all identified intersections~$\I{t}$ are valid.  
Moreover, we assume that we have observed the full intersection twice after time~$T_1$; let this time be~$T$.  
For the remainder of the proof, we consider only time stamps~$t > T$.

Let $w_t \in \mB_b$. This means $o_{t-1} > \Succ_K(w_t)$. Notice that if $w_t \in \I{t-1}$, then at time $t-1$, since the string $w_t$ is still available, the algorithm should have output a string that is at most $w_t$. Therefore, the reason the algorithm outputs~$o_{t-1} > \Succ_K(w_t)$ is that~$w_t \notin \I{t-1}$.  
Consequently,~$\I{t-1}$ is not consistent at time~$t$, since it does not contain~$w_t$.

By the algorithm, at time $t$, since $\I{t-1}$ is not consistent anymore, the identfied intersection at time $t$ should be a superset of $\I{t-1}$ in the descending chain $\mCh{t-1}$. In other words, suppose the descending chain of intersection at time $t-1$ is $\mCh{t-1} = \lanj{1} \supsetneq \lanj{2} \supsetneq \dots$, then there exists a $k \geq 1$ and $k < k'$ such that at time $t-1$, $\I{t-1} = \lanj{k'}$ and at time $t$, $\I{t} = \lanj{k}$. 

By the algorithm, at each step we either remain at the same intersection, move down by one level in the current descending chain of intersections, or move up within the current descending chain of intersections.
 Since $\I{t}$ is consistent at time $t$, it must have been consistent at any time up to time $t$, and thus always have been in the descending chain up to time $t$. Since we are after the second time when we return to a full intersection, and each time if we go down a chain we can go down by one level, there must be a time $t'$ where \begin{enumerate}
    \item $t' < t$, 
    \item $\I{t'} = \I{t} = \lanj{k}$, and 
\end{enumerate}
Define~$\phi(t)$ to be the largest such time stamp~$t'$.  
Clearly,~$t'$ occurs after the first return of a full intersection following~$T_1$.  
Also note that the map~$\phi(t)$, defined for all~$t$ with~$w_t \in \mB_b$ and~$t > T$, is injective, since~$\phi(t)$ identifies the most recent time stamp~$t'$ satisfying~$\I{t'} = \I{t}$.

Consider the algorithm’s output~$o_{t'}$.  
We claim that~$o_{t'} < w_t$.  
Indeed, both identified intersections at times~$t$ and~$t'$ are the same, that is,~$\I{t} = \I{t'}$, and~$w_t \in \I{t}$.  
By the design of the algorithm, at time~$t' < t$, it can output at most the smallest unused string in~$\I{t}$.  
Since~$w_t$ is first occupied at time~$t$ by the adversary, the smallest unused string in~$\I{t'} = \I{t}$ at time~$t'$ must be less than~$w_t$.  
Therefore,~$o_{t'} < w_t$.

Therefore, we have found an injective map $\rho: \{w_t: w_t \in \mB_g, t> T\} \to O$ by $\rho(w_t) = o_{\phi(t)}$. This map is thus injective with $o_{\phi(t)} < \rho(w_t)$.  
\end{proof}

We next show  Lemma \ref{claim:infrho}  implies Theorem \ref{thm:ldweak}. 

Consider the mapping $\bar\rho: C\setminus O \to O$ defined as follows.
For each $w_t \in C \setminus O$, if  $w_t \in \mB_g$, let $\bar\rho(w_t)= o_{t-1}$; if $w_t \in \mB_b$  and $t > T$ where $T$ is given by Lemma \ref{claim:infrho}, let $\bar\rho(w_t)= \rho(w_t)$. 

By the definition of~$\mB_g$ and Lemma~\ref{claim:infrho}, we have~$\bar\rho(w_t) \le \Succ_K(w_t)$ unless~$t \le T$.  
Furthermore, each~$o \in O$ has preimage size under~$\bar\rho$ of at most~$2$.  
Therefore, for any integer~$N$, let~$[N]_K$ denote the first~$N$ strings in~$K$, and let~$\mB = \mB_g \cup \mB_b$.  
Then we have
\[
|[N]_K \cap O| \ge \frac{|\mB_g \cap [N]_K| + |\mB_b \cap [N]_K|}{2} - T.
\]

In addition, since the adversary will eventually enumerate every string in $C$, we have 
\[ |[N]_K \cap O| + |[N]_K \cap \mB_g| + |[N]_K \cap \mB_b| \geq |[N]_K \cap C|. \]
This is because the three sets $O, \mB_g, \mB_b$ are mutually exclusive, and their union  cover $C$. 
Combining these two inequalities together, we have 
\[|[N]_K \cap O| + 2(|[N]_K \cap O| + T) \geq |[N]_K \cap C|. \]
Since $T$ is a fixed constant, and $C$ has lower density $\alpha$ in $K$, we have that as $N$ goes to infinity, $O\cap [N]_K$ has density at least $\alpha/3 - o(1)$ in $[N]_K$. Therefore the lower density of $O$ in $K$ is at least $\alpha/3$. 

\end{proof}

\subsubsection{Optimal bound}
\begin{thm}\label{thm:infRankHighLD2}
    There is an algorithm $\ma$ that achieves element-based generation in the limit in the partial enumeration model and has the following property. Given any countable collection of languages $\coll$ with  an underlying ordering of the strings in all the languages, and for any adversarial partial enumeration $E$ of one of the languages $\trueL \in \coll$ such that the adversary enumerates at least a lower density at least $\alpha>0$ fraction of the strings in $\trueL$, for any $\epsilon > 0$, the set of output strings $O(E,\ma)$ generated by the algorithm has a lower density in $K$ that is at least $\alpha/2$. 
\end{thm}
By Theorem \ref{thm:halfupperbound}, this theorem implies that we essentially achieve the best possible bound on the lower density. 

The key idea behind the improved bound is the following observation.  
In the previous argument, the reason we achieved only the bound~$1/3$ instead of~$1/2$ is that for each~$w_t \in C \setminus O$, we defined a mapping~$\bar\rho$ by  
\[
\bar\rho(w_t) =
\begin{cases}
o_{t-1}, & \text{if } w_t \in \mB_g,\\[4pt]
\rho(w_t), & \text{if } w_t \in \mB_b \text{ and } t > T,
\end{cases}
\]
where~$T$ is the time given by Lemma~\ref{claim:infrho}.  
Under this mapping, the preimage of each~$o$ has size at most~$2$.  
This preimage bound is the bottleneck in the argument, and it is unavoidable for the previous algorithm:  
at each time~$t$ when a string~$w_t \in \mB_b$ is produced by the adversary, it maps to a previous algorithm output~$o_{t'}$ for some~$t' < t$.  
However, the same output~$o_{t'}$ may also correspond to a string~$w_{t'} \in \mB_g$.  
Thus, each output~$o_{t'}$ can be counted twice, effectively reducing the achievable bound on~$|O|$ and thereby yielding the factor~$1/3$.

The new ingredient in the argument is to reduce the number of algorithm outputs~$o_{t'}$ that correspond simultaneously to elements of~$\mathcal{B}_b$ and~$\mathcal{B}_g$.  
To achieve this, each time the algorithm outputs~$o_t$, we conceptually treat it as if the algorithm had output not just a single string but a set of~$s$ consecutive strings.  
We refer to this set as a \emph{pod}. So conceptually, each string in $\mB_g$ corresponds to $s$ outputs, while each string in $\mB_b$ only corresponds to one, thus reducing the collapse.  However, notice that in reality the algorithm can only output one string at a time, so we need to be careful with the analysis. 
A more precise description of this construction is given below. 

\vspace{0.1in}
\noindent\textbf{Aggressive set, priority string list, and pods.}
Recall $\I{t}$ is the intersection of several languages identified by the algorithm $\Acc$ at time $t$, given by Algorithm 1 (described in the proof of Theorem \ref{thm:accurateoften}). 

At each time $t$, the language will decide on 
\begin{enumerate}
\item an \emph{aggressive set}~$\fbga{t}$, which consists of strings that the language guess believes could be part of the true language, though including them may risk overshooting;
    \item a \emph{pod} $P_t$, which contains the strings that have the priority to be output by the algorithm.
\end{enumerate}
These two definitions are the same as the previous section. 
In the algorithm description later, we will explain how we chose $\fbga{t}$ and the pods $P_t$ and $\mP_t$.

\vspace{0.1in}
\noindent\textbf{Algorithm with parameter $(s_t)_{t \geq 1}$.}
Fix a sequence of positive integers $s_t \geq 2$ for each $t \geq 1$. 

\begin{enumerate}
\item If $\I{t+1} = \I{t}$, do not fall back to any language. Let the pod $P_{t+1}$ be the set of $s_{t+1}$ strings consisting of the next $s_{t+1}$ unused (neither by the adversary nor the algorithm nor any of the previous pods) strings in $\I{t+1}$. 

Let $\mP_{t+1} = \bigcup_{t' \leq t+1} P_{t'}$. Pick the smallest string in $\mP_{t+1}$ and output that as $o_t$. 

\item If $\I{t+1}$ is a strict subset of~$\I{t}$, we make an aggressive guess by setting~$\fbga{t+1} = \I{t}$.  
By Algorithm~1, note that $\I{t+1}$ is a strict subset of~$\I{t}$ only when both sets belong to the descending chain of intersections at time~$t+1$, namely
\[
\mCh{t+1} = \lanj{1} \supsetneq \lanj{2} \supsetneq \lanj{3} \supsetneq \dots.
\]
Suppose $\I{t}$ corresponds to the set~$\lanj{i}$; then $\I{t+1}$ corresponds to the set~$\lanj{i+1}$.
    
    Set the pod $P_{t+1}$ as follows: 
    We will create a pod of size $s_{t+1}$. Let $P_{t+1}$ be the next $s_{t+1}$ strings haven't been used by algorithm nor adversary before which also did not appear in pods before. 
   
    Let $\mP_{t+1} = \bigcup_{t' \leq t+1} P_{t'}$. Pick the smallest string in $\mP_{t+1}$ and output that as $o_t$.

\item If $\I{t+1}$ is a strict superset of $\I{t}$. In this case we aggressively guess $\I{t+1}$, i.e., set $\fbga{t+1} = \I{t+1}$.   
Let the pod $P_{t+1}$ be the set of $s_{t+1}$ strings consisting of the next $s_{t+1}$ unused (neither by the adversary nor the algorithm nor any of the previous pods) strings in $\fbga{t+1}$.

Let $\mP_{t+1} = \bigcup_{t' \leq t+1} P_{t'}$. Pick the smallest string in $\mP_{t+1}$ and output that as $o_t$. 

\item \label{enum:alginf5} Else, we look for the minimum (in terms of set inclusion) identified intersection $I$ such that $I$ is a superset (i.e., precedes) $\I{t}$ in $\mCh{t}$ and is a superset (i.e., precedes) $\I{t+1}$ in $\mCh{t+1}$, if exists. Note that if there are identified intersection $I$ such that $I$ is a superset (i.e., precedes) $\I{t}$ in $\mCh{t}$ and is a superset (i.e., precedes) $\I{t+1}$ in $\mCh{t+1}$, then the minimum one (under set inclusion) must exist, since there are only finitely number of sets precedes $\I{t}$ in $\mCh{t}$ and finitely number of sets preceds $\I{t+1}$ in $\mCh{t+1}$. 

    \begin{enumerate}
        \item If no such $I$ exists, then do not fall back. 
        
        Let the pod $P_{t+1}$ be the set of $s_{t+1}$ strings consisting of the next $s_{t+1}$ unused (neither by the adversary nor the algorithm nor any of the previous pods) strings in $\I{t+1}$.

Let $\mP_{t+1} = \bigcup_{t' \leq t+1} P_{t'}$. Pick the smallest string in $\mP_{t+1}$ and output that as $o_t$. 
        
        \item
If such an $I$ exists, then aggressively guess $\fbga{t+1} = I$. 

Let the pod $P_{t+1}$ be the set of $s_{t+1}$ strings consisting of the next $s_{t+1}$ unused (neither by the adversary nor the algorithm nor any of the previous pods) strings in $\fbga{t+1}$.

Let $\mP_{t+1} = \bigcup_{t' \leq t+1} P_{t'}$. Pick the smallest string in $\mP_{t+1}$ and output that as $o_t$. 
    \end{enumerate}
\end{enumerate}

\begin{proof}[Proof of Theorem \ref{thm:infRankHighLD2}]
We will apply the algorithm where $s_t = t$ or for any growing sequence $(s_t)_{t \geq 1}$. 

Let $P$ be all the strings in all the pods at the end of the algorithm. 
Let $\mB_g$ be the set of ``good" strings in $C \setminus (P \cup O)$, defined as \[\mB_g = \{ w_t \in C \setminus (P \cup O): {\text{all elements in }} P_{t-1} {\text{ are at most }} \Succ_K(w_t) \}.\]
This set consists of adversary inputs for which the algorithm's previous output is not far behind. 

Let $\mB_b$ be the set of ``bad" strings, defined as \[\mB_b = C \setminus (P \cup O) \setminus \mB_g.\] 

Similar to Lemma \ref{claim:infrho}, we have the following lemma. 
\begin{lem}\label{claim:infrho2}
There exists a finite time~$T$ and an injective map
\[
\rho: \{\,w_t : w_t \in \mB_b,\, t > T\,\} \longrightarrow \{\,t' \in \mathbb{N} : P_{t'} \text{ is defined}\,\}
\]
such that every element of the pod ~$P_{t'} = \rho(w_t)$ is smaller than~$w_t$ in the underlying ordering of the strings.
\end{lem}  
\begin{proof}
Fix a positive integer~$s$.  
By Lemma~\ref{lem:property}, we may assume that we are beyond the finite time~$T_1$ after which all identified intersections~$\I{t}$ are valid, and that we have already observed the full intersection twice after time~$T_1$.  
Let this latter time be~$T$.  
For the remainder of the proof, we assume~$t > \max(T, s)$.

Let~$w_t \in \mB_b$.  
This means that some element of~$P_{t-1}$ is larger than~$\Succ_K(w_t)$.  
Note that since~$w_t \notin P \cup O$, the string~$w_t$ is unused by the algorithm, the adversary, and all previous pods at time~$t-1$.  
Thus, if~$\I{t-1}$ is consistent at time~$t$, then by the algorithm’s description we would not move down the descending chain of intersections from time~$t-1$ to~$t$.  
Moreover, consistency implies~$w_t \in \I{t-1}$; hence, the algorithm should have included~$w_t$ in the pod at time~$t-1$, since a larger string was included in that pod.  
Therefore, by the algorithm’s rules, the reason~$w_t$ was not included in the pod at time~$t-1$ must be that~$\I{t-1}$ is not consistent at time~$t$.

By the algorithm, at time $t$, since $\I{t-1}$ is not consistent anymore, the identfied intersection at time $t$ should be a superset of $\I{t-1}$ in the descending chain $\mCh{t-1}$. In other words, suppose the descending chain of intersection at time $t-1$ is $\mCh{t-1} = \lanj{1} \supsetneq \lanj{2} \supsetneq \dots$, then there exists a $k \geq 1$ and $k < k'$ such that at time $t-1$, $\I{t-1} = \lanj{k'}$ and at time $t$, $\I{t} = \lanj{k}$. 

By the algorithm, in each step we either stay at the same intersection, or choose the identified intersection by moving down {\it one level} in the current descending chain of intersections, or jumping up (one or even more levels) in the current descending chain of intersection (and discard the lower part of the chains and switch to a new tail of chains). Since $\I{t}$ is consistent at time $t$, it must have been consistent at any time up to time $t$, and thus always have been in the descending chain up to time $t$. Since we are after the second time when we return to a full intersection, and each time if we go down a chain we can go down by one level, there must be a time $t'$ where \begin{enumerate}
    \item $t' < t$, 
    \item $\I{t'} = \I{t} = \lanj{k}$, and 
\end{enumerate}
Define $\phi(t)$ be the largest such $t'$, then clearly  $t'$ is after the first return of a full intersection after $T_1$. Also notice that this map $\phi(t)$ for all $t$ with $w_t \in \mB_b$ and $t > T$ is injective, since $\phi(t)$ look for the most recent time stamp $t'$ with $\I{t'} = \I{t}$.

Consider the algorithm’s pod at time~$t'$, denoted by~$P_{t'}$.  
All elements of~$P_{t'}$ must precede~$w_t$ in the underlying ordering of strings.  
This is because the identified intersections at times~$t$ and~$t'$ are the same, that is,~$\I{t} = \I{t'}$, and~$w_t \in \I{t}$.  
By the algorithm’s design, at time~$t' < t$, it will include the $s_{t'}$ smallest unused strings in~$\I{t}$ to $P_{t'}$.  
Since~$w_t$ is first occupied by the adversary at time~$t$ and was never included in any previous pod, the smallest unused string in~$\I{t'} = \I{t}$ at time~$t'$ must be less than~$w_t$.  
Hence, every element of~$P_{t'}$ precedes~$w_t$.

We now show that~$\rho$ is injective.  
Fix~$t'$.  
For all times~$t$ such that~$\rho(t)$ is mapped to the same pod~$P_{t'}$, the previous argument implies that~$\I{t-1}$ is consistent at time~$t-1$ but not at time~$t$, and moreover that~$\I{t} = \I{t'}$.  
Suppose, for contradiction, that there exist~$t_1 < t_2$ such that~$\rho(t_1) = \rho(t_2) = P_{t'}$.  
By the definition of~$\rho$, we would then have~$\rho(t_2) = P_{t_1}$ instead of~$P_{t'}$, since~$t' < t_1 < t_2$.  
As all pods are disjoint, this yields a contradiction.  
Therefore,~$\rho$ is injective.
\end{proof}

We are now ready to show  Lemma \ref{claim:infrho2}  implies Theorem \ref{thm:infRankHighLD2}.

Since a pod will be created for each adversary input, and when $t \geq s$, $s_t \geq s$. Thus by the definition of $\mB_g$, we have 
\begin{equation}
    |P \cap [N]_K| \geq s|\mB_g \cap [N]_K|  - T -s^2. \label{eq:bg}
\end{equation}

By Lemma \ref{claim:infrho2}, we have that 
\begin{equation}
    |\mB_b \cap [N]_K| \leq |P \cap [N]_K| / s - O(T) - O(s^2)/s \label{eq:bb}
\end{equation}
since the injective mapping $\rho$ is to maps to individual pod, but each pod contains at least $s$ elements for $t \geq s$. 

For each pod, we claim that all its elements belong either to~$O$ or to~$C$.  
Suppose, for contradiction, that there exists an element~$a$ in this pod such that~$a \notin C$ and~$a$ is never output by the algorithm.  
Then, by the algorithm’s design, after this pod was created, all subsequent outputs of the algorithm must be strings smaller than~$a$, since the smallest unused element in the pods always has priority.  
However, there are only finitely many strings smaller than~$a$, which contradicts the fact that~$C$ is infinite and hence the algorithm runs for infinitely many time stamps.  
Therefore, every element in the pod must lie in~$O \cup C$.
Therefore, 
\begin{equation}
  P\setminus C \in O.  \label{eq:PinO}
  \end{equation}
Thus (\ref{eq:bg}) can be written as
\begin{equation}
    |O \cap [N]_K| + |P \cap C  \setminus O \cap [N]_K| \geq s|\mB_g \cap [N]_K|  - T -s^2 \label{eq:bg2}
\end{equation}
By the same reasoning, 
\begin{equation}
|\mB_b \cap [N]_K| \leq    ( |O \cap [N]_K| + |P \cap C  \setminus O \cap [N]_K| )/s - O(T)/s - O(s). \label{eq:bb2}
\end{equation}

We trivially have
\begin{equation}
|(P \cap C \setminus O) \cap [N]_K| + 
    |(P \cap O)\cap C \cap [N]_K| + |\mB_b \cap [N]_K| + |\mB_g \cap [N]_K|  = |C \cap [N]_K|, \label{eq:union}
\end{equation}  as $\mB = \mB_g \cup \mB_b$ and they are disjoint from $O \cup P$ by definition. 

Plugging in (\ref{eq:bb2}) and (\ref{eq:bg2}) into (\ref{eq:union}), we have 
\begin{align}
  |(P \cap C \setminus O) \cap [N]_K| + 
    |(P \cap O)\cap C \cap [N]_K| +  |O \cap [N]_K|/s + |P \cap C  \setminus O \cap [N]_K|/s - O(T)/s  \nonumber \\
    + |O \cap [N]_K|/s + |P \cap C  \setminus O \cap [N]_K|/s + T/s  + O(s) \geq |C \cap [N]_K|, \nonumber 
\end{align}
which implies 
\begin{align}
  (1+2/s)  |(P \cap C \setminus O) \cap [N]_K| + 
    |(P \cap O)\cap C \cap [N]_K| + (2/s) |O \cap [N]_K|/s 
      + O(T)/s + O(s) 
      \geq |C \cap [N]_K|,  \label{eq:unionnew}
\end{align}

We now claim \begin{equation}
    |P \cap O \cap [N]_K| \geq |P \cap C \setminus O \cap  [N]_K|. 
    \label{eq:withinP}
\end{equation}
Notice that every element in~$P \cap C$ will eventually be enumerated by the adversary.  
Hence, each element in~$P \cap C$ will either be output by the algorithm (and thus belong to~$O$) or be used exclusively by the adversary (and thus belong to~$C \setminus O$).
Thus, the set~$P \cap (C \setminus O)$ consists exactly of those elements in the pods that were generated by the adversary but never by the algorithm.  
Let~$w_t$ denote the adversary’s input at time~$t$, and suppose that~$w_t$ belongs to a pod~$P_{t'}$.  
Note that~$t' < t$; otherwise, since pods are constructed only from elements unused by either the adversary or the algorithm at the time of their creation, the pod could not have included~$w_t$.  
Hence, at time~$t$, the pod~$P_{t'}$ has already been created.  

Because~$w_t$ was not generated at time~$t-1$ and~$w_t \notin O$, it follows that the algorithm’s output at time~$t-1$ satisfies~$o_{t-1} < w_t$.  
Otherwise, by the algorithm’s description—since the pod~$P_{t'}$ already existed at time~$t-1 \ge t'$—the algorithm would have output~$w_t$ instead.  
Therefore, for each~$w_t \in P \cap (C \setminus O)$, there is a unique output element strictly smaller than~$w_t$ in the underlying ordering of strings.  
Since all algorithm outputs are drawn from the pods, we have that
$|P \cap O \cap [N]_K| \geq |P \cap C \setminus O \cap  [N]_K|$.

Plugging in (\ref{eq:withinP}) into (\ref{eq:unionnew}), we have
\begin{align*}
  (1+2/s)  |(P \cap O)\cap [N]_K| + 
    |(P \cap O)\cap C \cap [N]_K| + (2/s) |O \cap [N]_K|/s 
      + O(T)/s  + O(s)
      \geq |C \cap [N]_K|, 
      % \label{eq:union2} 
\end{align*}

This implies 
\begin{equation}
   (2+2/s)|(P \cap O)  \cap [N]_K| + 
     (2/s) |O \cap [N]_K|/s 
      + O(T)/s  + O(s) \geq  |C \cap [N]_K|. 
\end{equation}
Thus 
\begin{equation}
   (2+4/s)|O \cap [N]_K| + O(T)/s  + O(s) 
    \geq  |C \cap [N]_K|. \label{eq:union2} 
\end{equation}

 Since $T$ and $s$ are fixed constants, and $C$ has lower density $\alpha$ in $K$, we have that as $N$ goes to infinity, $O\cap [N]_K$ has density at least \[ \frac{\alpha}{2+4/s} \] in $[N]_K$. Therefore when $s$ is sufficiently large, the lower density of $O$ in $K$ is at least $\alpha(1/2 - \epsilon)$ for some fixed $\epsilon > 0$. As $s_t = t \to \infty$, we have the lower density of $O$ in $K$ is at least $\alpha / 2$. 
\end{proof}

\section{Identification with Partial Enumeration: A Topological Characterization}\label{sec:top}

Following the outline from the paper's introduction, we now turn to the classical Gold-Angluin model of {\em language identification in the limit}, where the goal is not simply to generate from the unknown language but to actually identify it.   We bring the partial enumeration perspective to this problem, asking what types of non-trivial guarantees we can give when the adversary need only enumerate a subset $C$ of the true language $K$.
It is too much to require that an algorithm identify the true language $K$ under these conditions, and so the natural analogue of the identification property in the case of partial enumeration is that the algorithm should produce a (finite representation of a) language $M_t$ in time step $t$ such that after some finite time, $M_t$ sits between $C$ and $K$: that is, $C \subseteq M_t \subseteq K$ for all sufficiently large $t$.
Note that this coincides with the standard definition of identification in the limit in the special case where $C = K$, i.e., when the adversary enumerates the full language.

We will give an exact characterization of when identification in the limit is possible in this partial enumeration model.
We will express this characterization in topological terms, where it is arguably particularly natural to formulate.

The paper~\cite{kleinberg2025density} was the first to draw a precise connection between formal language learning and topological constructions, interpreting the process of learning as the identification of a \emph{limit object} within a suitably defined topological space under the full enumeration model.  
Their work demonstrated that the convergence behavior of a learning algorithm can be understood through the topology induced by the space of hypotheses, opening a conceptual bridge between inductive inference and topology.

In this paper, we advance this perspective.  
Rather than treating the connection as a loose analogy, we develop a systematic topological framework for studying models of language learning.  
In particular, we show that the principal notions of identifiability correspond to well-defined \emph{separation axioms} in topology: for example full enumeration aligns with the $T_D$ property.  
This correspondence enables a unified language for reasoning about learnability, continuity, and robustness across different models of inference.

It also provides a more intuitive and conceptually cleaner characterization of when a countable collection of languages is identifiable especially in the partial enumeration model, making it easier to analyze structural properties of language classes.  
In particular,  this framework yields  new insights into the nature of learnability.  
For instance, we show that identifiability under partial enumeration is \emph{topologically robust} under finite removal of strings from the underlying alphabet—an invariance that fails dramatically in the full enumeration model.  
Moreover, the topological formulation renders this robustness naturally: the argument follows naturally from the preservation of $T_D$ separation under finite perturbations within our defined topology.  
These results reveal a deeper structural distinction between the two learning paradigms, allowing classical phenomena in inductive inference to be expressed and understood in purely topological terms.

\vspace{0.2in}
Our first result establishes this connection concretely by providing a topological reformulation of Angluin’s identification theorem in the full enumeration model, showing that learnability in this setting is equivalent to the $T_D$ separation property of the induced topology.

We now introduce the topological framework for this question.  
First, without loss of generality, we assume that each language in~$\coll$ is a subset of~$\mathbb{N}$, since any discrete countable set of strings can be enumerated by~$\mathbb{N}$.  
Given our countable collection of languages~${\mathcal{X}}$, and any infinite subset~$C' \subseteq \mathbb{N}$ that is contained in some language of~${\mathcal{X}}$, we define a topology~$\topo{C'}$ on~$C'$ as follows:

For each language $\lan \in \coll$ such that $C'\subset \lan$, and each finite subset $F$ of $C'$,
define one of the basic open sets to be
\[
U_{\lan, F, C'} = \{\lan' \in \mathcal{X} \mid F \subseteq \lan', |\lan' \cap C'| = \infty \}.
\]
For each language $\lan \in \coll$ such that $C'\not\subset \lan$, and each finite subset $F$ of $C'$,
define one of the basic open sets to be
\[
U_{\lan, F, C'} = \{\lan\}
\]
the singleton set. 

Let the collection of sets \( U_{\lan,F, C'} \), ranging over all languages \( \lan \in\mathcal{X} \)  and all finite subsets $F \subset C'$, serve as a basis for the open sets. This defines a topology $\topo{C'}$ on \( \mathcal{X} \).  

To see it is an open cover, notice that if we have some language $\lan' \in U_{\lan_1,F_1, C'} \cap U_{\lan_2,F_2, C'}$, then clearly $\lan' \in U_{\lan',F_1 \cup F_2, C'}$ and $U_{\lan',F_1 \cup F_2, C'} \subset U_{\lan_1,F_1, C'} \cap U_{\lan_2,F_2, C'}$. Also notice that for any language $L$, trivially $L \in U_{L, \emptyset, C'}$. 

The next definition and result establish the formal connection between language learning and the notion of \emph{limit points} in a topology, making precise how convergence of hypotheses corresponds to identification in the limit.
\begin{defn}
We say that a language $L$ is a {\it true language candidate} with respect to $C'$ if the following hold:
\begin{enumerate}\label{def:truelanguagecandidate}
    \item $C' \subset L$
\item There exists a sequence of languages~$L_1, L_2, \dots$ such that each~$L_i \ne L$, each~$L_i \cap C$ is infinite, and there exists an enumeration~$a_1, a_2, \dots$ of~$C$ satisfying  
\[
a_1 \in \bigcap_{i \ge 1} L_i, \quad 
a_2 \in \bigcap_{i \ge 2} L_i, \quad \text{and in general, } 
a_n \in \bigcap_{i \ge n} L_i \text{ for all } n \ge 1.
\]
\end{enumerate}
\end{defn}

\begin{claim}\label{claim:truelanguagelimit}
    If $L$ is a {\it true language candidate} with respect to $C'$ if and only if $L$ is a limit point in $\topo{C'}$. 
\end{claim}
\begin{proof}
    Suppose~$L$ is a \emph{true language candidate} with respect to~$C'$.  
Then there exists a sequence of languages~$L_1, L_2, \dots$ satisfying Item~2.  
We will show that~$L$ is a limit point of the sequence~$L_1, L_2, \dots$ in~$\topo{C'}$.  

Indeed, every basic open set in~$\topo{C'}$ consists of all languages that contain some finite subset~$F \subset C'$.  
Let~$t_F$ be the smallest index such that
\[
F \subseteq \bigcup_{1 \le i \le t_F} \{a_i\}.
\]
Then, by Item~2, we have~$L_{t_F} \in U_{L, F, C'}$.  
Since~$L_{t_F} \ne L$, it follows that~$L$ is a limit point of the sequence in~$\topo{C'}$.

We now show the converse.  
Suppose that~$L$ is a limit point in~$\topo{C'}$.  
In particular, this means that the singleton set~$\{L\}$ is not open in~$\topo{C'}$.  
Hence,~$C' \subset L$.  

Let~$a_1, a_2, \dots$ be an arbitrary enumeration of~$C'$, and define~$F_n = \{a_1, a_2, \dots, a_n\}$.  
Then the basic open set~$U_{L, F_n, C'}$ always contains some language~$L_n \ne L$, since~$L$ is a limit point in~$\topo{C'}$.  
In particular, each~$L_n$ contains~$F_n = \{a_1, a_2, \dots, a_n\}$.  
Therefore,~$L$ is a true language candidate for the sequence~$L_1, L_2, \dots$.
\end{proof}

In the partial enumeration model, it is not meaningful to define ``identification in the limit'' as identifying the true language~$K$.  
Indeed, suppose that~$K$ and another language~$L$ have intersection~$C$, and the adversary fully enumerates~$C$.  
In this case, the algorithm can never distinguish between~$K$ and~$L$, since it only observes elements of~$C$.
\begin{defn}[Identification in the Limit in Partial Enumeration Model]
    In the partial enumeration model, suppose the adversary enumerates an infinite subset $C$ of the true language $K$. We say the algorithm {\it identifies in the limit} if there is a finite time $T$ after which, the algorithm always outputs a set of strings fully contains $C$ but is contained in $K$.  
\label{def:partial-ident}
\end{defn}

It is useful to talk about the representational issues involved in this Definition \ref{def:partial-ident}, which were previewed earlier.
This definition is formulated in terms of an algorithm that is allowed to output an \textbf{arbitrary infinite} set as one of its hypotheses.  
Although such an algorithm is clearly not computationally feasible, it provides a convenient idealization for delineating the theoretical boundary between what is learnable and what is not.  
Perhaps surprisingly, this idealization \textbf{does not} increase the learner’s power.  
In all of our positive results, we show that whenever identification is achievable by such an unrestricted algorithm, it is in fact achievable by an algorithm that outputs only a \emph{finite intersection} of its candidate languages.  
Conversely, our negative results apply even to algorithms that can output arbitrary infinite sets.  
In other words, allowing arbitrary infinite hypotheses, no matter how powerful this may seem, does not confer any additional learning capability beyond that obtainable with finite intersection representations.

So once an algorithm for identification in the partial enumeration model is allowed to produce finite intersections of candidate languages in $\mathcal{X}$, it does not really gain in power if we allow it to produce more elaborate representations.  
On the other hand, algorithms that can produce finite intersections have a big benefit in representational power over languages that can only output individual languages in $\mathcal{X}$ (as in the standard model with full enumeration by the adversaries). 
Indeed, it is straightforward to check that if the algorithm can only output individual languages in $\mathcal{X}$, there is very little non-trivial that one can say in the partial enumeration model.

The following example succinctly illustrates some of the basic gains in representational power that comes from allowing finite intersections in the partial enumeration model.

\begin{example}
A further illustration of the fragility of the full enumeration model arises when the learner is restricted to outputting a single language rather than intersections of hypotheses.  
Let the set of distinct primes be $p_1, p_2, \dots$.
Consider the collection of languages
\[
    L_i = \{\, n \in \mathbb{N} \mid p_i \text{ divides } n \,\}
    \quad \text{for each prime } p_i.
\]
Each $L_i$ consists of the multiples of a prime $p_i$.  
Suppose an adversary enumerates all multiples of some fixed composite number $m = p_i p_j$.  
In this case, every element enumerated belongs to both $L_i$ and $L_j$, and hence is consistent with multiple hypotheses simultaneously.  
If the learner is restricted to outputting a single language, it has no finite way to distinguish whether the target is $L_i$, $L_j$, or perhaps $L_r$ for some other prime divisor $p_r$ of $m$, as the learner will risk from overshooting.   
Therefore, no algorithm that produces only a single language as output can successfully identify this class in the limit.

By contrast, if the learner is allowed to output \emph{intersections} of hypotheses—effectively representing uncertainty over finitely many consistent languages—then it could output $L_i \cap L_j$ while the data remains ambiguous, and refine its hypothesis once disambiguating evidence appears.  
This shows that allowing intersection-based hypotheses (as in the partial enumeration model) strictly increases the expressive power of the learner, making it possible to achieve identifiability where the single-output restriction fails.\end{example}

\begin{claim}\label{claim:limitnoidentification}
If~$L$ is a true language candidate with respect to~$C'$ for a sequence of languages~$L_1, L_2, \dots$ that do not fully contain~$C'$, then no algorithm can identify~$L$ in the limit.
\end{claim}
\begin{proof}
If~$L$ is a true language candidate with respect to~$C'$, then there exists a sequence of languages~$L_1, L_2, \dots$ satisfying the required condition.  
Let the adversary begin by pretending that the true language is~$L_1$, enumerating elements from the infinite set~$L_1 \cap C'$.  
Suppose the algorithm can identify in the limit.  
Then, after some finite time, it must conclude that the enumeration set is a subset of~$L_1$.  

At this point, since~$L$ is a true language candidate, there must exist another language~$L_2'$ that contains a string in~$C' \setminus L_1'$, contains all adversary outputs up to this point, and satisfies~$|L_2' \cap C'| = \infty$.  
The adversary can then pretend that the true language is~$L_2'$ and continue enumerating accordingly.  
After some finite time, the algorithm should again infer that the true language is a subset of~$L_2'$.  

Repeating this process, we can construct a sequence of languages~$L_1', L_2', \dots$, where each~$L_i'$ is consistent with all previously enumerated strings but contains at least one new string not in~$L_1' \cup \dots \cup L_{i-1}'$.  
The adversary then pretends the true language is~$L_i'$ by enumerating that new string together with other elements from~$C' \cap L_i'$, which is infinite.  
Since no~$L_i'$ fully contains~$C'$, the algorithm can never stabilize on a hypothesis that fully contains~$C'$.  
Hence, identification in the limit is impossible in this case.
\end{proof}

\begin{remark}
Unlike the full enumeration model, where each language is assumed to be infinite, in the partial enumeration model the intersection~$L \cap C$ may be finite.  
This introduces additional subtleties in handling such finite intersections.

The requirement that~$L_i \cap C$ be infinite is essential.  
If~$|L_i \cap C| < \infty$, then the statement no longer holds.  
For example, let~$L_i = \{1, 2, \dots, i\}$ for each~$i \ge 1$, and let~$L_0 = \mathbb{N}$ with~$C = \mathbb{N}$.  
Even though this sequence satisfies all the conditions in Definition~\ref{def:truelanguagecandidate} except that~$|L_i \cap C| < \infty$, we claim that there still exists an algorithm that identifies in the limit.  

The algorithm operates as follows: at each time~$t$, it always guesses~$L_0$ unless the adversary stops enumerating.  
At that point, the algorithm correctly identifies the true language based on the adversary’s inputs up to that time.  
Clearly, this algorithm identifies in the limit.
\end{remark}

We consider the intersection topology~$\mathcal{T}$ obtained by intersecting all topologies~$\topo{C}$, where~$C$ ranges over all infinite subsets of the languages.  
Equivalently, a set~$U$ is open in~$\mathcal{T}$ if and only if it is open in some~$\topo{C'}$ for an infinite sublanguage~$C'$.  
This is the coarsest topology with respect to all~$\topo{C}$.  
It is well known that any limit point in each of the individual topologies~$\topo{C}$ is also a limit point in their intersection topology~$\mathcal{T}$.  
In general, the converse need not hold; however, we show that, surprisingly, in our setting the converse does hold.

\begin{lem}\label{lem:limitpointequiv}
The set of limit points in the intersection topology~$\mathcal{T}$ is exactly the union of the sets of limit points in the individual topologies~$\topo{C'}$.
\end{lem}
\begin{proof}
    It suffices to show that there is no language $L$ which is a limit point in $\mathcal{T}$ but not a limit point in $\topo{C'}$ for all infinite sublanguages $C'$. 

First, observe that if there exists another language~$L' \ne L$ such that~$L' \cap L$ is infinite, then~$L'$ is a limit point in~$\topo{L'\cap L}$.  
Indeed, for any basis element of~$\topo{L'\cap L}$, say 
\[
U = \{\,L'' \in \coll \mid F \subset L'',\ |L'' \cap (L' \cap L)| = \infty\,\},
\]
we have~$F \subset L' \cap L$, which implies~$F \subset L$.  
Hence,~$L' \in U$, and thus both~$L$ and~$L'$ are limit points in~$\topo{L'\cap L}$.  
In other words,~$L$ serves as a true language candidate for the constant sequence~$L', L', L', \dots$.  
This yields a contradiction, since we assumed that~$L$ is not a limit point in~$\topo{C'}$ for any infinite sublanguage~$C'$.

Thus, we may assume that for any~$L' \ne L$, the intersection~$L \cap L'$ is finite.  
Then, for any infinite sublanguage~$C'$, if~$C'$ is not a subset of~$L$, the language~$L$ is an isolated point in~$\topo{C'}$ by the definition of~$\topo{C'}$.  
Hence, it remains only to consider the case where~$C' \subset L$.

Since we have assumed that for any~$L' \ne L$, the intersection~$L' \cap L$ is finite, it follows that the singleton~$\{L\}$ is an open set in the topology~$\topo{C'}$ for any infinite sublanguage~$C' \subset L'$ and any~$L' \ne L$.  
However, we cannot make~$\{L\}$ an open set in the intersection topology~$\mathcal{T}$, as this would imply that~$L$ is not a limit point in~$\mathcal{T}$.  

We have already shown that for any~$C'$ not contained in~$L$, the set~$\{L\}$ is open in~$\topo{C'}$.  
Thus, the only possibility for~$\{L\}$ not being open in the intersection topology is when~$C'$ is an infinite subset of~$L$.  
In particular, there must exist some infinite~$C' \subset L$ such that~$\{L\}$ is not open in~$\topo{C'}$.  
This means that for every finite subset~$F \subset C'$, the basic open set~$U_{L, F, C'}$ must contain a language other than~$L$ having infinite intersection with~$C'$; otherwise,~$\{L\}$ would be a singleton open set, which contradicts our assumption.  

Let~$a_1, a_2, \dots$ be an enumeration of the infinite set~$C' \subset L$, and define~$F_n = \{a_1, a_2, \dots, a_n\}$.  
Then there exists some language~$L_n \ne L$ in~$U_{L, F_n, C'}$.  
By our previous argument, this implies that~$L$ is a true language candidate with respect to~$C'$ for the sequence~$L_1, L_2, \dots$.  
Hence, by Lemma~\ref{lem:limitpointequiv},~$L$ is a limit point in~$\topo{C'}$, yielding a contradiction.
\end{proof} 

This topology incorporates \emph{all possible} limit-point scenarios.  
The open sets encode the information corresponding to the worst possible adversarial manipulation.  
In a sense, the open sets precisely capture the most adversarial situation under which the learner attempts to infer the true language.  

To illustrate this, consider the following toy example.  
Let~$L_0 = \mathbb{Z}$.  
For each even index~$n = 2k$ with~$k > 1$, define  
\[
L_n = \{\,1 \le i \le k\,\} \cup 2\mathbb{N},
\]
and for each odd index~$n = 2k - 1$ with~$k > 0$, define  
\[
L_n = \{\, -k \le i \le -1\,\} \cup (-2\mathbb{N}).
\]
If the adversary behaves \emph{non-adversarially} and chooses the enumeration set~$C$ containing~$\{0\}$, then~$L_0$ is the only language in the open set corresponding to any~$F$ that contains~$\{0\}$.  
(In other words, the learner can correctly identify~$L_0$ as soon as~$\{0\}$ is revealed.)  

On the other hand, if the adversary behaves adversarially and chooses~$C' \subset \mathbb{N} \setminus \{0\}$, then for any~$F \subset C'$, the language~$L_0$ belongs to~$U_F$ together with many other languages whose indices exceed the largest element of~$F$.

In other words,~$L_0$ is not an isolated point in~$\topo{C'}$, but rather a limit point in this topology.  
A similar argument holds for~$C' \subset -\mathbb{N} \setminus \{0\}$.  
Hence, in the intersection topology,~$L_0$ is not an isolated point but a limit point, even though it is isolated in the individual topology when~$C$ contains~$\{0\}$.  

Thus, the intersection topology effectively disregards non-adversarial scenarios, capturing only the worst-case adversarial behaviors.  
The open sets in the intersection topology are of the form
\[
L_0 \cup \{L_{2k-1} : k \ge k_1\} \cup \{L_{2k} : k \ge k_2\},
\quad \text{for all } k_1, k_2 \ge 1.
\]
This structure is dramatically different from that of the open sets in the full enumeration model:  
because of the intersection topology, it simultaneously records all possible ways the learner may converge to~$L_0$ under the worst adversarial scenarios for~$C$.

\

Recall that two points in a topological space are \emph{topologically indistinguishable} if every open set containing one also contains the other.  
This notion of indistinguishability induces an equivalence relation on the space.  
The quotient space obtained by identifying all topologically indistinguishable points is called the \emph{Kolmogorov quotient}; in this quotient, every two distinct points are topologically distinguishable.
\begin{lem}\label{lem:distinct}
    Our topological space $\mathcal{T}$ is homeomorphic to its Kolmogorov quotient if and only no two languages $L, L'$ satisfying $L \cap L'$ being infinite.  
\end{lem}
\begin{proof}
If two points $L, L'$ are indistinguishable in $\topo{C'}$ for some infinite sublanguage $C'$, then $L, L'$ are indistinguishable in the intersection topology, since any open set in $\topo{C'}$ either contain both or contain neither one of them. Therefore if $L \cap L'$ is infinite, then choose $C' = L \cap L'$. Then by the definition of $\topo{C'}$, every open set containing $L$ will also contain $L'$, and vice versa. Therefore if there are two languages $L, L'$ satisfying $L \cap L'$ being infinite, then $L, L'$ are indistinguishable in the intersection topology. 

Now we prove the converse. Suppose no two languages have infinite intersection. We show every two languages are topologically distinguishable. For this we just need to find an open set $U$ that is open in every individual topology $\topo{C'}$ that separates $L$ and $L'$. 

Consider the set  $U = \bigcup_{w \notin L'} \{\lan' \mid w\in \lan'\}$. Clearly $U$ does not contain $L'$, because each set $\{\lan' \mid w\in \lan'\}$ does not include $L'$ as string $w$ is not in $L'$. We claim that $U$ is an open set in each topology $\topo{C'}$ where $C'$ is an infinite sublanguage of some language. Since no two languages  have infinite intersection, there is exactly one language $J$ that contains $C'$. The other languages are in their singleton open sets. If $J$ is not $L'$, then $\bigcup_{w \in J \setminus L'} \{\lan' \mid w\in \lan'\}$ is an open set in $\topo{C'}$. For each other language $J'$ which contains an $w \notin L'$ but not  in $\bigcup_{w \in J \setminus L'} \{\lan' \mid w\in \lan'\}$, clearly this language does not fully contain $C'$ since $J$ is the unique language containing $C'$, so $J'$ is in its singleton open set. The union of open set is again open. So we have shown $U$ is open in $\topo{C'}$ if $J$ is the unique language containing $C'$ and $J \neq L'$. Now we assume $J$ is $L'$. Then again, since any language $J'$ contains some $w \notin L'$ does not fully contain $C'$, $\{J'\}$ is an open set in  $\topo{C'}$. The union of such $J'$ is again open. Thus $U$ is open in $\topo{C'}$ in all cases of $C'$. Therefore we have shown that $U$ contains $L$ but not $L'$. And thus $L$ and $L'$ are not topologically equivalent. 
\end{proof}

Note that even though for any language $L$, there is an uncountable number of choices for its infinite sublanguage $C$, but we do not have to study all the possible uncountable number of $C$. The theorems above show that we could study all the $C$'s at once. Furthermore, the theorem indicates the importance of studying the intersection of languages. In particular, any two languages which have infinite intersection will be indistinguishable in the partial enumeration model. 

We just saw that if two languages have infinite intersection, then they are indistinguishable in the intersection topology. 

Given any finite (possibly empty set $F$), we define $\topo{F}$ be the intersection topology of the previous $\topo{C'}$ where $F \subset C'$. The same proof will yield the following, as the only properties of $C'$ we've used is that in each individual topology $\topo{C'}$, $C'$ is an infinite subset of some language. 
\begin{lem}
    Lemmas \ref{lem:limitpointequiv} and \ref{lem:distinct} also hold for $\topo{F}$ for each finite $F$.
\end{lem}

Let $a_1, a_2, \dots$ be the adversary enumeration sequence. We thus define the sequence of topologies $\mathcal{T}_1, \mathcal{T}_2, \dots$ where $\mathcal{T}_i = \bigcap_{\{a_1, \dots, a_i\} \in C} \topo{C}$. 

\subsection{A Topological Approach to Identification in the Limit}

As a first step toward the partial enumeration model, we begin by considering the classical Gold-Angluin model of identification in the limit when the adversary fully enumerates the unknown language $K$.
In this context, which we cover in this subsection, we are able to give a clean, novel topological characterization of when identification in the limit is possible, essentially reformulating and proving Angluin's characterization \cite{angluin1979finding,angluin1980inductive} by topological means.

Let $\mathcal{X}$ be a class of languages on a countable set of underlying strings $\Sigma$, and let $f: \Sigma^* \to \mathcal{C}$ be a learner (i.e., a function mapping finite sequences of strings to hypotheses).

\begin{thm}[Angluin \cite{angluin1980inductive}] Suppose the set of underlying strings is countable. 
 The collection of languages 
 is identifiable in the limit in the full enumeration model if and only if for every language $L$, there exists a finite set $D_L \subseteq L$ (called a \textit{tell-tale set}) such that:
\[
\forall L' \in \mathcal{C},\quad D_L \subseteq L' \Rightarrow L' \not\subseteq L
\]

\noindent
That is, no strict subset of $L$ in $\mathcal{C}$ contains $D_L$.
\end{thm}

Let us use $\mathcal{T}_{\rm{full}}$ to denote the topological space where the open sets are generated by the basis: $U_F = \{\lan \mid F \subset \lan\}$.  The language learning problem is closely related to separation properties of the space, i.e., how can we distinguish two points in the space by open sets (and closed sets).

\begin{thm}[Topological restatement of the Angluin Theorem]\label{thm:topAngluin}
Suppose the set of underlying strings is countable. 
    In the Gold-Angluin full enumeration model, the identification in the limit is possible if and only if (1) the set of languages is countable, and (2) in $\mathcal{T}_{\rm{full}}$ 
    every point $L$ satisfies \( \overline{\{L\}} \setminus \{L\} \) is closed. In other words, this is $T_D$ space. 
\end{thm}

We begin by recalling some basic definitions in topology.  

A topological space~$(X, \tau)$ is called a \emph{$T_0$ space} (or \emph{Kolmogorov space}) if, for every pair of distinct points~$x, y \in X$, there exists an open set~$U \in \tau$ such that  
\[
(x \in U \implies y \notin U) \quad \text{or} \quad (y \in U \implies x \notin U).
\]
Equivalently, in terms of the \emph{specialization preorder}~$\leq$,  
\[
x \leq y \text{ and } y \leq x \;\Rightarrow\; x = y.
\]

\medskip
\noindent
The \emph{specialization preorder}~$\leq$ on~$X$ is defined by  
\[
x \leq y 
\quad \text{if and only if} \quad 
(\forall\, U \in \tau)\, [\,x \in U \Rightarrow y \in U\,].
\]
That is, every open set containing~$x$ also contains~$y$.

\medskip
\noindent
A topological space~$(X, \tau)$ is called a \emph{$T_1$ space} if, for every pair of distinct points~$x, y \in X$, there exists an open set containing~$x$ but not~$y$, and vice versa.  
Equivalently, a space is $T_1$ if all singleton sets~$\{x\}$ are closed.

\medskip
\noindent
A topological space~$X$ is called a \emph{$T_D$ space} if, for every point~$x \in X$,  
the set~$\overline{\{x\}} \setminus \{x\}$ is closed.  
This means that each point can be \emph{separated from the other points in its closure} by an open set.  
The $T_D$ condition lies strictly between the $T_0$ and $T_1$ separation axioms.

\medskip
\noindent
There is a well-known hierarchy among these separation axioms:
\[
T_1 \;\Rightarrow\; T_D \;\Rightarrow\; T_0.
\]
These implications are strict: there exist $T_0$ spaces that are not $T_D$, and $T_D$ spaces that are not $T_1$.

\medskip
\noindent
In general, the \emph{separation axioms} describe how distinctly a topological space can differentiate between points using open or closed sets.

\begin{proof}[Proof that Angluin condition is equivalent to that $\mathcal{T}_{\rm{full}}$ being an $T_D$ space]
This is a restatement of the classical Angluin result. In the Angluin result, we need a finite set $D_L$ of strings such that no language that is a proper subset of $L$ contains $D_L$. 

Suppose the collection of language satisfies Angluin's condition. Notice that $L' \leq L$ if and only if $L' \in \overline{\{L\}}$. So Angulin's condition implies that the open set $U_{D_L}$ does not intersect any $L' \leq L$. Also notice $L' \leq L$ if and only if $L' \in \overline{\{L\}}$.

So Angulin's condition is equivalent to say that there is an open set $U_F$ for some finite $F$ containing $L$ which is disjoint from $B = \operatorname{cl}(\{x\}) \setminus \{x\}$. 
We show this is equivalent to say  $B$ is closed for every $L$. 

Indeed, if $B$ is closed for every $L$, consider $U_L$ to be the complement of $B$, which is open since $B$ is closed. Since each open set $U_L$ will contain a base element, it will contain some $U_F$ for some finite $F$. This open set $U_F$ is a subset of the complement of $B$, so it is disjoint from $B$. 

Now we assume that there is such an open set $U_F$  containing $L$ but is disjoint from $B$ for every language $L$. Since $\operatorname{cl}(\{x\}) $ is closed, $\operatorname{cl}(\{x\}) \cap (U_F)^c$ is again closed. However, the complement $(U_F)^c$ does not include $L$ and includes every point in $\operatorname{cl}(\{x\}) \setminus \{x\}$. So $\operatorname{cl}(\{x\}) \cap (U_F)^c = \operatorname{cl}(\{x\}) \setminus \{x\}$, which is a closed set.  
\end{proof}

\begin{proof}[Proof of Theorem \ref{thm:topAngluin}]

We first show that if the collection of languages is uncountable, then identification in the limit is impossible.  
Suppose, for contradiction, that there exists an algorithm~$A$ which guarantees identification in the limit.  
For each language~$L$, let~$E_L$ denote some adversarial enumeration of~$L$ such that, after some finite time~$T$, the algorithm~$A$ stabilizes at~$L$.  
This implies that there exists a finite subset of strings~$W_L$ such that~$A(W_L) = L$.  

However, the set of all finite subsets of a countable set is itself countable, whereas the set of languages~$L$ is uncountable.  
Hence,~$A$ cannot be surjective, contradicting the assumption that~$A(W_L) = L$ for all~$L$.  
Therefore, no algorithm can guarantee identification in the limit when the collection of languages is uncountable.

We now show that if~$\mathcal{T}_{\rm{full}}$ is a~$T_D$ space and is countable, then there exists an algorithm that identifies in the limit.  
This condition is equivalent to requiring that each point~$L$ has an open set~$U_L$ that is disjoint from~$\overline{\{L\}} \setminus \{L\}$.  
Without loss of generality, we may assume that each~$U_L$ is a basis element.  
In particular, for any~$L'$ with~$L < L'$, the open set~$U_{L'}$ does not contain~$L$.  

We define the algorithm as follows.  
While the adversary enumerates strings in the true language~$L$, the algorithm monitors whether the set of enumerated strings includes the finite subset determining~$U_L$ (since~$U_L$ is assumed to be a basic open set) and whether this set of strings is contained in~$L$.  
Whenever both conditions hold, the algorithm guesses that the true language is~$L$.  
If multiple such languages~$L$ satisfy these conditions, the algorithm selects the one with the smallest index.

To see why the algorithm works, first observe that for any~$L < L'$, the open set~$U_{L'}$ does not contain~$L$.  
Therefore, regardless of how the adversary enumerates the strings, the algorithm’s guesses will never enter the region corresponding to~$U_{L'}$.  
In particular, the algorithm will never guess a language~$L'$ with~$L' > L$.  

On the other hand, the open set~$U_L$ is separated from all~$L'$ with~$L' < L$.  
Hence, after some finite time—once the enumerated strings fall within the finite set defining~$U_L$—the algorithm will no longer guess any language~$L'$ with~$L' < L$.  
The only remaining candidates are the languages~$L''$ that are incomparable with~$L$ under the specialization preorder.  
For any such~$L''$ whose index is smaller than that of~$L$, the algorithm will exclude~$L''$ as soon as a string in~$L \setminus L''$ is enumerated by the adversary.  
Consequently, after some finite time,~$L$ will be the language with the smallest index satisfying all conditions, and the algorithm will stabilize on~$L$.

We now show that if the space is not a~$T_D$ space, then identification in the limit is impossible.  
The failure of the~$T_D$ property means that for some language~$L$, the set~$\overline{\{L\}} \setminus \{L\}$ is not closed.  
Equivalently,~$L$ is a limit point of the set~$\{L' : L' < L\}$.  
We now demonstrate that this implies identification in the limit cannot be achieved.

Pick any open set~$U_1$ containing~$L$.  
Then there exists some~$L_1 \ne L$ with~$L_1 \in U_1$.  
The adversary pretends that the true language is~$L_1$.  
Since~$L_1 < L$, this is valid, and after some finite time~$T_1$, the algorithm will output~$L_1$.  
Let~$E_1$ be the finite set of strings enumerated by the adversary up to time~$T_1$.  
The set~$U_{E_1}$ is then an open set containing both~$L_1$ and~$L$.

Now choose a string~$w_1 \in L \setminus L_1$, and consider the open set~$U_1 \cap U_{\{w_1\}}$, which does not contain~$L_1$.  
Define
\[
U_2 = U_{E_1} \cap U_1 \cap U_{\{w_1\}}.
\]
This open set contains~$L$ but not~$L_1$, and every language in~$U_2$ is consistent with~$E_1$.  
By assumption,~$U_2$ must contain another language~$L_2 < L$.  
The adversary now pretends that the true language is~$L_2$; after some finite time~$T_2$, the algorithm outputs~$L_2$.  
Let~$E_2$ be the set of strings enumerated by the adversary up to time~$T_2$ (so~$E_1 \subset E_2$).  
Choose~$w_2 \in L \setminus L_2$, and define
\[
U_3 = U_{E_2} \cap U_2 \cap U_{\{w_2\}}.
\]
Then~$U_3$ contains~$L$ but not~$L_1$ or~$L_2$, and all languages in~$U_3$ are consistent with~$E_2$.  
By the same reasoning,~$U_3$ must contain another language~$L_3 < L$.  

Repeating this construction inductively, we obtain a sequence of languages~$L_1, L_2, \dots, L_n, \dots$ with~$L_i < L$ and corresponding time stamps~$T_1 < T_2 < \dots$, such that at each time~$T_i$, the algorithm outputs~$L_i$.  
Since each~$U_{i+1}$ remains open and nonempty, this process continues indefinitely, producing an infinite sequence of distinct outputs.  
Thus, the algorithm never stabilizes, contradicting the assumption that identification in the limit is possible.
\end{proof}

We prove a slightly stronger statement.  
Note that identification in the limit may occur even if the algorithm is not aware that it has correctly identified the true language.  

We say that a collection of languages~$\mathcal{X}$ is \emph{separated in the limit} if, for all distinct~$L, L' \in \mathcal{X}$ and for any adversarial enumeration~$(w_0, w_1, w_2, \ldots)$ of (possibly a subset of)~$L$, there exists~$n_0 \in \mathbb{N}$ such that for all~$n \ge n_0$, the learner knows that at least one of~$L$ or~$L'$ is not the true language.

Separation in the limit implies identification in the limit.  
Indeed, consider the following reasoning.  
We begin by comparing the two languages with the smallest indices.  
After some finite time, separation in the limit ensures that at least one of them can be eliminated.  
We then compare the remaining candidate with the language having the next smallest index, and continue this process iteratively.  
Eventually, the true language~$K$ will be compared against all others, and since~$K$ can never be eliminated, it will remain as the final candidate.  
If, when distinguishing between two potential languages, the algorithm always ``guesses'' that the true language is the one with the smaller index, then this procedure guarantees identification in the limit.

\begin{cor}\label{cor:T1}
    A collection of language is separation in the limit if and only if the topological space is $T_1$.  
\end{cor}
\begin{proof}
    Suppose there is $L_1 \leq L_2$ for some $L_1 \neq L_2$.  Then any open set containing $L_1$ also contains $L_2$. In particular, given the adversary enumeration $(w_0, w_1, w_2, \ldots)$, and any finite set $F_n = \{w_1, w_2, \dots, w_n\}$, the open set $U_{F_n}$ will contain both of $L_1, L_2$ if it contains $L_1$. Thus we will never able to separate $L_1$, $L_2$, and thus separation in the limit is impossible. 

    Now suppose for any $L_1$ and $L_2$, there is an open set $U$ containing $L_1$ but not $L_2$. Then $U$ will contain some finite open set $U_F$. So if the enumeration $(w_0, w_1, w_2, \ldots)$ does not contain $F$, then we know $L_1$ is not the correct language. If the enumeration contains $F$, then at the point where all the strings in $F$ are revealed, we are able to separate $L_1$ and $L_2$. 
\end{proof}

\begin{cor}
    In the full enumeration model, separation in the limit is possible if and only for any two languages $L_1 \neq L_2$, no language is a superset of the other. 
\end{cor}
\begin{proof}
    First assume no language is a superset of the other. For $L_1 \neq L_2$, we show there exists an open set that contains $L_1$ but not $L_2$. This is clear since we can choose $w \in L_1 \setminus L_2$, which is possible since $L_2$ is not a superset of $L_1$. Then the open set $U_w$ contains $L_1$ but not $L_2$. 

    Now we assume that $L_1 \subsetneq L_2$. Then we show any open set containing $L_1$ also contains $L_2$. For this it suffices to show any base open set containing $L$. Since $L_1 \subset L_2$, any base open set $U_F$ containing $L_1$ means $F \subset L_1$, but it also implies $F \subset L_2$. Thus $U_F$ also contains $L_2$. 

    Thus we have shown that the space is $T_1$ if and only if no language is a subset of the other. By Corollary \ref{cor:T1} we are done. 
\end{proof}

\subsection{Extending to the partial enumeration model}\label{subsec:toppartial}

We now turn back from the Gold-Angluin setting with full enumeration to the case of an adversary that may only partially enumerate the set $K$ --- that is, it can choose any infinite subset $C$ of $K$ to enumerate.

In this partial enumeration model, the situation becomes significantly more complex, since multiple topologies are involved—one for each possible choice of~$C'$.  
In particular, the number of such choices of~$C'$ is uncountable.  
At first glance, this may seem to make the identification problem impossible, as there are only countably many basis open sets.

\begin{thm}\label{thm:identificationpartial}
    Identification in the limit is possible in the partial enumeration model if and only if for any $L$ and any $C \subset L$, the Kolmogorov quotient of the space $\topo{C}$ is a $T_D$ space. 
\end{thm}

Here recall the definition of specialization preorder. In this case, it is equivalent to: $L_1 \leq L_2$ if $|L_1 \cap C|, |L_2 \cap C| = \infty$ and $L_1 \cap C \subset L_2 \cap C$. 
The proof of this theorem turns out to be much more involved compared to the full enumeration modal.

% {\color{red}
Before proceeding to the proof, we make an important remark.  
 It is essential to emphasize the significant technical difference between \emph{partial enumeration} and \emph{full enumeration}.  
At first glance, each topology $\topo{C}$ may appear to resemble the full enumeration model restricted to the sublanguage $C$.  
However, the crucial difficulty and difference lies in the fact that the learner \textbf{does not know in advance} which $C$ is relevant, and there are uncountably many possible choices of $C$.  
Hence, one cannot simply apply techniques from the full enumeration setting to each $\topo{C}$ individually.  
This distinction is not merely technical: as shown in~\cite{kleinberg2025density}, when the collection of languages is uncountable, the class may fail even to be \emph{generatable in the limit}.  
Thus, the partial enumeration model represents a \textbf{genuinely different regime} of learnability, where uncertainty over the underlying context $C$ fundamentally alters what can be identified.
% }

\begin{proof}[Proof of Theorem \ref{thm:identificationpartial}]
    We first prove the easier direction: if identification in the limit is possible, then  for any $L$ and any $C \subset L$, the space $\topo{C}$ is a $T_D$ space.
To see why this is true, suppose the adversary is enumerating~$C \subset L$, where~$L$ is the true language.  
Effectively, this setting is equivalent to identification in the limit in the full enumeration model, where each language~$L'$ is replaced by~$L' \cap C$.  
If~$\topo{C}$ is not a~$T_D$ space, then there exists a language~$J$ such that for every open set~$U$ in~$\topo{C}$, there is another language~$J' \in U$ satisfying~$(J' \cap C) \subsetneq (J \cap C)$.  
Note that~$|J' \cap C| = \infty$.  

In other words, the set~$\overline{\{J\}} \setminus \{J' : J' \cap C = J \cap C\}$ is open, and thus~$J$ is a limit point.  
By Claims~\ref{claim:limitnoidentification} and~\ref{claim:truelanguagelimit}, it follows that identification in the limit is impossible.

We now prove the converse: if each space~$\topo{C}$ is a~$T_D$ space, then identification in the limit is possible.  
In retrospect, however, this condition alone may sound to be too weak to guarantee identification in the limit.  
Indeed, in the full enumeration model, if the adversary enumerates an infinite set~$C$, then the~$T_D$ property of~$\topo{C}$ ensures identification in the limit: each superset of the true language~$C$ has its own ``certificate''—an open set that separates it from all of its proper subsets.  

The caveat in the partial enumeration model is that the algorithm does \emph{not} know which~$C$ is being enumerated.  
Among the uncountably many possible subsets~$C$, the learner has no way of knowing which~$\topo{C}$ to use to locate the appropriate certificate.  
By contrast, in the full enumeration model, each language~$L$ corresponds to a unique choice of~$C$, namely~$C = L$ itself, so each language admits a unique certificate of identification.  
In the partial enumeration model, however, each language~$L$ may be associated with uncountably many possible subsets~$C$—the uncountably many infinite subsets of~$L$—making the identification task fundamentally more complex.

To address this issue, we need to find a ``certificate-free" way of identification. We define our algorithm as follow: 

At time $t$, only considering the languages which are consistent up to time $t$, i.e., the adversary enumeration $w_1, \dots, w_t$ are all strings in the language. Write these languages from left to right depending on their original labels in the list of languages, say $\Lc{1}{t},  \Lc{2}{t}, \dots$. 
Write $\Ic{1}{t} = \Lc{1}{t}$, $\Ic{2}{t} = \Lc{1}{t}\cap \Lc{2}{t}$, and in general, for $i \geq 1$, write 
\begin{equation}
    \Ic{i}{t} = \bigcap_{j=1}^i \Lc{j}{t}. 
\end{equation}
If some $\Lc{i}{t}$ has finite cardinality for some $i$, output the largest $k$
 such that $\Lc{k}{t}$ has infinite cardinality; if every $\Lc{i}{t}$ has infinite cardinality and $\bigcap_{j\geq 1} \Lc{j}{t}$ has infinite cardinality, then output $\bigcap_{j\geq 1} \Lc{j}{t}$; if every $\Lc{i}{t}$ has infinite cardinality but $\bigcap_{j\geq 1} \Lc{j}{t}$ has finite cardinality, then output $\bigcap_{1 \leq j\leq t} \Lc{j}{t}$. 

We first show that after some finite time, all the remaining languages coming before $K$ will contain $C$. This is because if some $L$ does not contain $c \in C$, then once $c$ is output by the adversary, $L$ will be eliminated. Therefore after some finite time, all the remaining languages whose ordering in the original language list are before $K$ all fully contain $C$. Therefore after some finite time $T$, our algorithm will always output a subset of strings of $K$. 

We now show that this algorithm also identifies in the limit. This is because, as $\topo{C}$ is a $T_D$ space, there is an open set $U$ that separates $K$ from the rest of the languages $L$ where $L \cap C \subsetneq C$, and $|L \cap C| = \infty$. This open set must contain a base open set of the form $U_F$ for some finite subset $F \subset C$. Then once all the strings in $F$ were enumerated by the adversary, all the remaining consistent languages, i.e., all the languages in $U_F$ will satisfy $(L \cap C) \supset C$, except for the languages $L$ containing $F$ but $|L \cap C| < \infty$. By our algorithm, since we will try to output the intersection of as many consistent languages as possible where their intersection is infinite, it means we will never include a language $L$ where $|L \cap C| < \infty$. Therefore the intersection should fully contain $C$. 
\end{proof}

In the special case of full enumeration, our argument provides an \textbf{alternative, certificate-free proof}of Angluin’s classical characterization of language identification in the limit.  
Unlike Angluin’s original proof, our approach does not require knowledge of the specific tell-tale sets for each language.  
Instead, the topological $T_D$ property itself guarantees their existence implicitly, offering a geometric and conceptually simpler interpretation of learnability.

Note that any collection of languages that is identifiable in the partial enumeration model is also identifiable in the full enumeration model. The converse is not true.

\subsection{Robustness of Partial Enumeration, and illustrative examples}

An additional noteworthy property is that identification by partial enumeration is robust under finite deletions from the underlying set of strings. In other words, removing finitely many strings from the domain does not affect whether a class is identifiable in the limit with partial enumeration. By contrast, this robustness fails for full enumeration: finite deletions can destroy identifiability (a class identifiable before deletion may cease to be so), though not the other way around. Below is the restatement of Theorem \ref{thm:introRobustness}. 
\begin{thm}
    Let $\mathcal{X}$ be a countable collection of languages, and let $\mathcal{X}'$ be obtained from $\mathcal{X}$ by removing a finite set of strings from the underlying ground set (thus removing those strings from each language and eliminating any duplicate languages that may result). Then $\mathcal{X}$ is identifiable in the limit with partial enumeration if and only if $\mathcal{X}'$ is identifiable in the limit with partial enumeration. 
\end{thm}
\begin{proof}
Let $W$ be the finite set of strings removed from the underlying alphabet. 
    Recall that for each infinite sublanguage $C'$, the topology $\topo{C'}$ on $\mathcal{X}$ is generated by the basic open sets
    \[
        U_{\lan, F, C'} = 
        \begin{cases}
            \{\lan' \in \mathcal{X} \mid F \subseteq \lan',\ |\lan' \cap C'| = \infty\}, & \text{if } C' \subseteq \lan, \\[4pt]
            \{\lan\}, & \text{if } C' \not\subseteq \lan,
        \end{cases}
    \]
    where $F$ ranges over all finite subsets of $C'$.
    Identifiability in the limit with partial enumeration is characterized by the requirement that, for every infinite $C'$, the Kolmogorov quotient of $\topo{C'}$ is a $T_D$ space by Theorem~\ref{thm:identificationpartial}.

    Define the natural map
    \[
        f \colon \mathcal{X} \to \mathcal{X}', \qquad 
        f(\lan) = \lan \setminus W.
    \]
    Since $W$ is finite, for each infinite $C'$ with $C' \cap W = \emptyset$, we have a bijective correspondence between the basic open sets of $\topo{C'}$ and those of $\topo{C'\setminus W}$:
    \[
        f(U_{\lan, F, C'}) = U_{f(\lan), F, C'\setminus W}.
    \]
    Hence $f$ induces a homeomorphism between $\topo{C'}$ on $\mathcal{X}$ and $\topo{C'\setminus W}$ on $\mathcal{X}'$. 
    In particular, the Kolmogorov quotient of $\topo{C'}$ is a $T_D$ space if and only if the Kolmogorov quotient of $\topo{C'\setminus W}$ is a $T_D$ space.
    Therefore, if $\mathcal{X}$ is identifiable in the limit with partial enumeration, then so is $\mathcal{X}'$.

  For the converse, assume that $\mathcal{X}'$ is identifiable in the limit with partial enumeration. 
    Suppose, for contradiction, that $\mathcal{X}$ is not. 
    Then there exists some infinite $C'$ for which $\topo{C'}$ fails to be $T_D$. 
    Thus, there is a language $L \in \mathcal{X}$ that is a limit point of an infinite sequence $\{L_i\} \subseteq \mathcal{X}$ in $\topo{C'}$, meaning that no basic open neighborhood of $L$ can separate it from all preceding points in the specialization preorder.

    Consider now the finitely many languages among the $L_i$ that differ from $L$ only on the finite set $W$. 
    Removing these finitely many exceptions does not affect the existence of a limit sequence: since $W$ is finite, the set of $L_i$ differing from $L$ on $W$ is also finite.
    The remaining $L_i$ form an infinite sequence for which $L$ remains a limit point with respect to $\topo{C'\setminus W}$ in $\mathcal{X}'$.
    But this contradicts the assumption that $\mathcal{X}'$ is identifiable in the limit with partial enumeration, since that would imply that $\topo{C'\setminus W}$ is $T_D$ and thus admits no such nontrivial limit point.
    Therefore, $\mathcal{X}$ must also be identifiable in the limit with partial enumeration.
\end{proof}

\paragraph{Failure of robustness in the full enumeration model.}

On the contrary, in the \emph{full enumeration model}, removing even a finite number of strings may destroy identifiability. 
For instance, consider the collection of languages
\[
    L_i = \mathbb{N} \setminus \{i\} \quad (i \geq 1),
    \qquad L_0 = \mathbb{N} \cup \{-1\}.
\]
This collection is identifiable in the limit under full enumeration, since $L_0$ contains the unique distinguishing element $-1$.  
In the full enumeration model, the adversary is required to eventually enumerate every element of the target language; hence, if the true language is $L_0$, the learner will eventually observe the special symbol $-1$ and can then conclusively identify $L_0$.

However, if we remove the single string $-1$ from the underlying alphabet (that is, let $W = \{-1\}$), we obtain the modified collection
\[
    L_i' = \mathbb{N} \setminus \{i\} \quad (i \geq 1),
    \qquad L_0' = \mathbb{N}.
\]
This new collection $\{L_i' : i \ge 0\}$ is \emph{not} identifiable in the limit under full enumeration.  
Intuitively, $L_0'$ and the tail of $\{L_i'\}$'s are now in some sense ``indistinguishable": there is no finite stage of enumeration at which the learner can determine whether a missing element $i$ will eventually appear or not.  
Topologically, this corresponds to the failure of even the $T_1$ separation property (as in this case $L_i'$ is always a subset of $L_0'$): the language $L_0'$ cannot be separated from all the $L_i'$ by any open set in the enumeration topology.  
Thus, while identification by partial enumeration is robust under finite deletions, identifiability in the full enumeration model is not: it can be lost, though never gained, by removing finitely many strings.

\vspace{0.3in}

We now illustrate more examples to show that it is more difficult for a set of languages to be identifiable with partial enumeration. 
% }

\begin{example}
    For each $i \geq 1$ let $L_i = \mathbb{N} \setminus \{i\}$, and $L_0 = \mathbb{N}$. In this example, the collection of languages is not identifiable in the full enumeration model, and thus is not identifiable in the partial enumeration model. The reason why it not identificable in the full enumeration model is that there is no open set that separates $L_0$  from all the other $L_i$'s simultanueously. 
\end{example}

\begin{example}
We now look at an example very similar to the previous one. 
    For each $i \geq 1$ let $L_i = (\mathbb{N} \setminus \{i\}) \cup \{-2\}$, and $L_0 = \mathbb{N} \cup \{-1\}$. Even by just pending two special strings $-2, -1$, this collection of languages becomes identifiable in the full enumeration model. This is because $L_0$ now is separated from the rest of the $L_i$'s by the special string $-1$. Since in the full enumeration model, all the strings need to be revealed, so $-1$ will be revealed eventually if $L_0$ is the true language. Another way to see it is that the specialization preorder is simply an antichain. However, this set of languages is not identifiable in the partial enumeration model. This is because by choosing $C = \mathbb{N}$,  the quotient space of $\topo{C}$ is not a $T_D$ space, as $\mathbb{N}$ is a limit point of languages $L_i$'s and $L_i \leq L_0$ in the specialization preorder for each $i \geq 1$. On the other hand, if $C$ contains either $-1$ or $-2$, then it becomes identifiable. But since the algorithm does not know which is the correct $C$, unless it sees $-1$ or $-2$, it can make the wrong guess infinitely often. 
    This also indicates that, this example is less robust in the sense that one can easily break the property that it is identifiable in the limit with full enumeration by deleting a finite number of strings ($-1 $ and $-2$). 
\end{example}

\begin{example}
   Consider the set of languages where for each $i \geq 1$, $L_i = \mathbb{N} \setminus \{i\}$. This set of language is more robust to be identifiable in the limit with full enumeration compared to the previous example in the sense that by removing a finite set of strings (and thus removing duplicate languages), it is still identifiable in the limit in the full enumeration model. However, this set again is identifiable in the full enumeration model since it is a $T_D$ space. However, this is again not identifiable in the partial enumeration model, even though the union of the "core" of each $L_i$'s is not a subset of one language.  Pick $C = 2\mathbb{N} \subset L_1$. Then again $\topo{C}$ is not a $T_D$ space since $L_1$ is a limit point of all the other languages whose specialization preorder is less than $L_1$ in the quotient of  $\topo{C}$. 
\end{example}

\begin{example}
    If the set of languages is finite, then clearly it is always identifiable in both models, since in our setting, any Kolmogorov quotient of a finite space is a $T_D$ space. 
\end{example}

\begin{example}
The following example is an non-trivial collection of languages that is identifiable in the partial enumeration model.  
Consider the family $L_i = \mathbb{N} i$, that is, the set of integer multiples of $i$.  
This family is identifiable because the only information required is the greatest common divisor (gcd) of the observed values in $C$.  
At each time step $t$, the algorithm can compute the gcd of all adversarial inputs seen up to time $t$, and identify all values that are multiples of this gcd.  
After some finite time $T$, this running gcd will coincide with the true gcd of all elements in $C$.  
Hence, from time $T$ onward, the algorithm can correctly output all multiples of this greatest common divisor.

  We can also extend this example to all arithmetic progressions.  
For each pair of positive integers $i, d$, define the language
\[
L_{i,d} = \{\, i + k d : k \in \mathbb{N} \,\},
\]
that is, the set of positive integers forming an arithmetic progression with initial term $i$ and common difference $d$.  
The key observation is that the algorithm only needs to determine the value of $d$ and the smallest representative element $i$ in $L_{i,d}\cap C$.  
At each time step $t$, the algorithm computes the greatest common divisor (gcd) of all adversarial inputs observed up to time $t$, and identifies all values equal to the smallest adversarial input seen so far plus multiples of this gcd.  
After some finite time $T$, this running gcd will coincide with the true gcd of all elements in $C$, and the smallest element of $C$ will also have appeared.  
Hence, from time $T$ onward, the algorithm can correctly identify a set that contains $C$.  
It is straightforward to verify that this set is always a subset of the true language.

\end{example}

\section{Conclusion}

We have considered language generation in the limit in a model with {\em partial enumeration}, in which an adversary only enumerates an infinite subset $C$ of the true language $K$. Note that such an adversaries now has an uncountable set of strategies at its disposal --- one for each subset of a language in the collection $\mathcal{X} = \{L_1, L_2, L_3, \ldots\}$ --- rather than just the countable set of strategies of an adversary that must choose a language in $\mathcal{X}$. 

We have shown that there is an algorithm that can achieve language generation in the limit even with this model of partial enumeration; and moreover, when the set $C$ the adversary produces has lower density $\alpha$ in the true language $K$, then the algorithm will produce an output set of lower density at least $\alpha/2$, and this bound is sharp. Specializing to the case of full enumeration from prior work, where $C = K$, this establishes a tight lower density of $1/2$, resolving the central open question from prior work on density in language generation. 

Our algorithm works by representing its hypothesized language at every step as a finite intersection of languages in $\mathcal{X}$, and we further show that with finite intersections as a representation, we can characterize when an algorithm is able to achieve {\em language identification in the limit} --- in the sense of the classical model of Gold and Angluin --- with an adversary performing partial enumeration. Our characterization is topological, and as a by-product of our analysis, we give a novel, succinct topological formulation of Angluin's classical characterization theorem for language identification in the limit, showing that it is equivalent to a topological space defined in \cite{kleinberg2025density} having the $T_D$ separation property. 

As a further direction, it would be interesting to ask what further results could be achieved with the topological methods developed here, and whether there are other clean characterizations of language learning properties that are expressible in terms of the topological spaces we work with. It would also be interesting to consider whether there are other strengthenings of the adversary's power that naturally reflect the properties faced by algorithms in real language generation scenarios, and whether they could be captured by the types of analysis considered in this work. 

\bibliographystyle{plain}
\bibliography{AIbib}

\end{document}